\documentclass[11pt]{article} 
\usepackage[numbers]{natbib}
\usepackage{booktabs} 
\usepackage[utf8]{inputenc} 

\usepackage{geometry} 
\usepackage{graphicx} 

\usepackage{booktabs} 
\usepackage{array} 
\usepackage{paralist} 
\usepackage{verbatim} 
\usepackage{subfig} 

\usepackage{fancyhdr} 
\usepackage{sectsty}
\allsectionsfont{\sffamily\mdseries\upshape} 

\usepackage[nottoc,notlof,notlot]{tocbibind} 
\usepackage[titles,subfigure]{tocloft} 

\usepackage[linesnumbered,noend,ruled,noline]{algorithm2e} 

\SetAlFnt{\small}
\SetAlCapFnt{\small}
\SetAlCapNameFnt{\small}
\SetAlCapHSkip{0pt}
\IncMargin{-\parindent}

\usepackage{amsmath}
\usepackage{amssymb}
\usepackage{amsthm} 
\usepackage{url}

\usepackage{color}
\usepackage{bm}
\usepackage{tikz}
\usepackage{thmtools, thm-restate}

\usepackage{hyperref}

\newtheorem{thm}{Theorem}
\newtheorem{theorem}[thm]{Theorem}
\newtheorem{lemma}[thm]{Lemma}
\newtheorem*{lemma*}{Lemma}

\newtheorem{claim}[thm]{Claim}

\newtheorem{corollary}[thm]{Corollary}

\newtheorem{observation}[thm]{Observation}

\usepackage{authblk}

\newcolumntype{C}[1]{>{\centering\arraybackslash}m{#1}}

\newcommand{\be}{\begin{equation}}
\newcommand{\ee}{\end{equation}}
\newcommand{\beq}{\begin{equation*}}
\newcommand{\eeq}{\end{equation*}}

\newcommand{\argmax}{\mathop{\rm argmax}}

\newcommand{\AutoAdjust}[3]{\mathchoice{ \left #1 #2  \right #3}{#1 #2 #3}{#1 #2 #3}{#1 #2 #3} }
\newcommand{\Xcomment}[1]{{}}

\newcommand{\InParentheses}[1]{\AutoAdjust{(}{#1}{)}}
\newcommand{\InBrackets}[1]{\AutoAdjust{[}{#1}{]}}
\newcommand{\Ex}[2][]{\operatorname{\mathbf E}_{#1}\InBrackets{#2}}
\newcommand{\Exlong}[2][]{\operatornamewithlimits{\mathbf E}\limits_{#1}\InBrackets{#2}}
\newcommand{\Prx}[2][]{\operatorname{\mathbf{Pr}}_{#1}\InBrackets{#2}}

\newcommand{\eqdef}{\overset{\mathrm{def}}{=\mathrel{\mkern-3mu}=}}
\newcommand{\vect}[1]{\ensuremath{\mathbf{#1}}}

\newcommand\restr[2]{{
  \left.\kern-\nulldelimiterspace 
  #1 
  \vphantom{\big|} 
  \right|_{#2} 
  }}

\def\prob{\Prx}

\newcommand{\opt}{\textsf{OPT}}

\renewcommand{\emptyset}{\varnothing}

\newcommand{\feasible}{\mathcal{F}}
\newcommand{\dsets}{\textsf{Disjoint-Maximal-Sets}}
\newcommand{\ksystem}{\textsf{Knapsack}} 

\newcommand{\dd}{\: \mathrm{d}}

\def \reals {{\mathbb R}}
\def \natural {{\mathbb N}}

\newcommand{\val}{v}
\newcommand{\vals}{\vect{\val}}
\newcommand{\valsmi}[1][i]{\vals_{\text{-}#1}}
\newcommand{\vali}[1][i]{{\val_{#1}}}

\newcommand{\valc}{\breve{v}}
\newcommand{\valcs}{\vect{\valc}}

\newcommand{\valh}{\hat{v}}

\newcommand{\valhi}[1][i]{{\valh_{#1}}}

\newcommand{\wal}{\widetilde{v}}
\newcommand{\wals}{\vect{\wal}}
\newcommand{\walsmi}[1][i]{\wals_{\text{-}#1}}
\newcommand{\wali}[1][i]{{\wal_{#1}}}

\newcommand{\pvals}{\vect{v}'}
\newcommand{\pvalsmi}[1][i]{\pvals_{\text{-}#1}}
\newcommand{\pvali}[1][i]{{v'_{#1}}}

\newcommand{\price}{p}
\newcommand{\prices}{\vect{\price}}
\newcommand{\pricei}[1][i]{{\price_{#1}}}

\newcommand{\distone}{D}
\newcommand{\dist}{\mathbf{D}}
\newcommand{\dists}{\vect{\dist}}

\newcommand{\disti}[1][i]{{D_{#1}}}
\newcommand{\distsmi}[1][i]{\dists_{\text{-}#1}}

\newcommand{\pdists}{\vect{\dist}'}

\newcommand{\pdistsmi}[1][i]{\pdists_{\text{-}#1}}

\newcommand{\supp}{\textsf{supp}}

\newcommand{\ind}[1]{\mathbb{I}\left[\vphantom{\sum}#1\right]}
\newcommand{\event}{\mathbb{I}_{\mathcal{E}(S)}}
\newcommand{\ev}{\mathcal{E}}

\newcommand{\oxi}{\overline{\xi}}

\newcommand{\maxsets}{\mathcal{S}}

\newcommand{\setsize}[1]{\lvert #1 \rvert}

\newcommand{\bidders}{N}
\newcommand{\outcomes}{\mathcal{F}}
\newcommand{\auc}{\textsc{auc}}
\newcommand{\auco}{\textsc{auc}_o}

\newcommand{\vbench}{\Exlong[\vect{\val}\sim \vect{\dist}]{\text{max}_{F \in \outcomes}\left\{\sum_{i \in F}{v_i}\right\}}}

\newcommand{\uprice}[1][g]{\textsc{u-price}_{#1}}

\newcommand{\cauct}{single-price clock auction}
\newcommand{\caucts}{single-price clock auctions}

\newcommand{\hide}[1]{}

\newcommand{\highcore}{\mathsf{HIGH\mbox{-}CORE}}
\newcommand{\hightail}{\mathsf{HIGH\mbox{-}TAIL}}
\newcommand{\lowcore}{\mathsf{LOW\mbox{-}CORE}}
\newcommand{\lowtail}{\mathsf{LOW\mbox{-}TAIL}}

\newcommand{\valciS}{\valc_{i,S}}
\newcommand{\valhiS}{\valh_{i,S}}

\newcommand{\activebidders}{A}

\begin{document}
\title{Bayesian and Randomized Clock Auctions\thanks{
The first author is partially supported by the European Research Council (ERC) under the European Union's Horizon 2020 research and innovation program (grant agreement No. 866132), and by the Israel Science Foundation (grant number 317/17).
The second and fourth authors were supported in part by NSF grants CCF-2008280 and CCF-1755955.
The third author is supported in part by Science and Technology Innovation 2030 ``New Generation of Artificial Intelligence'' Major Project No.(2018AAA0100903), Innovation Program of Shanghai Municipal Education Commission, Program for Innovative Research Team of Shanghai University of Finance and Economics (IRTSHUFE) and the Fundamental Research Funds for the Central Universities, and by RFIS grant No. 62150610500. 
We would like to thank Aviad Rubinstein for his feedback that helped us generalize our results.}}


\author[a]{Michal Feldman\thanks{michal.feldman@cs.tau.ac.il}}
\author[b]{Vasilis Gkatzelis\thanks{gkatz@drexel.edu}}
\author[c]{Nick Gravin\thanks{nikolai@mail.shufe.edu.cn}}
\author[b]{Daniel Schoepflin\thanks{drs332@drexel.edu}}
\affil[a]{Tel Aviv University}
\affil[b]{Drexel University}
\affil[c]{ITCS, Shanghai University of Finance and Economics}
\date{} 
\setcounter{Maxaffil}{0}
\renewcommand\Affilfont{\itshape\small}

\maketitle
\begin{abstract}
In a single-parameter mechanism design problem, a provider is looking to sell some 
service to a group of potential buyers. Each buyer $i$ has a private value $v_i$ for receiving this service, and some feasibility constraint restricts which subsets of buyers can be served simultaneously. Recent work in economics introduced  (deferred-acceptance) \emph{clock auctions} as a superior class of auctions for this problem, due to their transparency, simplicity, and very strong incentive guarantees. 
Subsequent work in computer science focused on evaluating these auctions with respect to their social welfare approximation guarantees, leading to strong impossibility results:
in the absence of prior information regarding the buyers' values, no deterministic clock auction can achieve a bounded approximation, even for simple feasibility constraints with only two maximal feasible sets.

We show that these negative results can be circumvented either by using access to \emph{prior information} or by leveraging \emph{randomization}. 
In particular, we provide clock auctions that give a $O(\log\log k)$ approximation for general downward-closed feasibility constraints with $k$ maximal feasible sets, for three different information models, ranging from full access to the value distributions to complete absence of information. 
The more information the seller has, the simpler and more practical these auctions are.
Under full access, we use a particularly simple deterministic clock auction, called a {\em single-price clock auction}, which is only slightly more complex than posted price mechanisms. In this auction, each buyer is offered a single price, then a feasible set is selected among those who accept their offers. 
In the other extreme, where no prior information is available, this approximation guarantee is obtained using a complex randomized clock auction.
In addition to our main results, we propose a parameterization that interpolates between single-price clock auctions and general clock auctions, paving the way for an exciting line of future research.

\end{abstract}

\thispagestyle{empty}
\clearpage
\pagenumbering{arabic}

\section{Introduction}
Our goal in this paper is to design auctions for the following well-studied class of mechanism design problems: given some set $N$ of $n$ buyers that request a service, and a feasibility constraint $\feasible \subseteq 2^{N}$ that restricts the subsets of buyers that can be served simultaneously, we need to decide which feasible subset of buyers $F\in\feasible$ should be served, and how much each served buyer should pay for the service. A crucial obstacle is that we do not know the \emph{value} $\vali$ of each buyer $i$ for the service, i.e., the amount that they are willing to pay for it. Therefore, unless our auction is carefully designed, the buyers can choose to misrepresent this value, aiming to minimize their payment.

For a simple concrete example, consider a setting where some ticket broker is looking to sell 100 tickets for a concert and each potential buyer is interested in purchasing a specific number of tickets so that they can go with their family, or group of friends. The buyers would go to the concert only if they are able to secure their desired number of tickets (depending on their group size), but the broker has just 100 tickets, so it would only be feasible to satisfy subsets of buyers whose total demand for tickets is at most 100. 
The broker needs to use some mechanism that decides who should get tickets and how much they should each pay, and there are several different types of mechanisms to choose from. For example, she could ask each potential buyer to privately submit a bid with the amount that they are willing to pay, and then use this information to decide the outcome. 
A much simpler, but possibly inefficient, option would be to just post a fixed price and sell the tickets at that price on a first-come-first-served basis.

\citet{MS2014,MS2019} recently provided a compelling argument that an ideal solution for this class of problems is a (deferred-acceptance) \emph{clock auction}, and they designed one for reallocating radio frequency licenses, generating almost 20 billion dollars in revenue~\cite{fcc}.
An ascending clock auction takes place over a sequence of rounds: each buyer is offered a \emph{personalized} price that weakly increases over time and if the price exceeds the amount they are willing to pay, they can permanently drop out of the auction and pay nothing (note that this is unlike other ascending auctions that assign prices to items instead). 
When the set of buyers that remain active (the ones that have not dropped out) becomes feasible, the auction can terminate and serve these buyers at the cost of the last price they accepted.
Therefore, designing a clock auction reduces to the algorithmic problem of carefully choosing the vector of prices offered to buyers in each round.

Any mechanism that follows the clock auction format automatically satisfies a list of very appealing properties that make clock auctions a highly practical solution that is well-suited for real-world applications. For example, not only are they \emph{strategyproof} (meaning buyers have no reason to unilaterally misrepresent their value); they are actually \emph{obviously strategyproof}~\cite{L2017}, which implies significantly stronger incentive guarantees, including \emph{weak group-strategyproofness}, i.e., that even if a coalition of buyers misrepresented their values in a coordinated way, they would not all benefit from this deviation. Furthermore, clock auctions guarantee \emph{(i) transparency} (there is no way in which the auctioneer can mishandle, behind the scenes, the information provided by the buyers), \emph{(ii) unconditional winner privacy} (the winners of the auction never need to reveal their true value), and \emph{(iii) simplicity} (the buyers do not need to understand the inner workings of the auction; all they need to know is the price offered to them in each round). These are properties that most strategyproof auctions do not satisfy (see Section~\ref{sec:properties} for additional discussion on these properties).

Motivated by the strong evidence in favor of clock auctions, subsequent work focused on analyzing their performance in a variety of different settings (e.g., \cite{DGR2017,GMR2017,DTR2017,LM2020,K2015,bichler2020strategyproof,GPPS21,BGGST22,CGS2022}). 
For the problems studied in this paper, these efforts were overshadowed by a strong impossibility result: even for seemingly very simple instances, all prior-free deterministic clock auctions are bound to perform poorly~\cite{DGR2017}. However, the restriction to prior-free auctions, i.e., ones where the seller has no prior information regarding the amount that the buyers may be willing to pay, is rather unrealistic, given the vast amounts of historical data being gathered and stored nowadays. 
In fact, the standard model in auction design is in the \emph{Bayesian setting}, where the value $\vali$ of each buyer $i$ is drawn independently from some distribution $\disti$ and, although the seller does not know the realization of the values, she knows the distributions.
Also, although deterministic auctions are more appealing in practice, using randomization to overcome adversarially constructed instances can be a very effective tool. Despite the long literature on Bayesian auctions and randomized auctions, the power of these tools in the context of clock auctions remains largely unexplored.
In this paper, we address this gap in the literature by analyzing both deterministic Bayesian clock auctions and randomized prior-free clock auctions. 

\subsection{Our Results}\label{sec:results}
We study the problem of maximizing the social welfare, i.e., the sum of the values of the served buyers, and we compare the expected welfare of our auctions to the expected optimal welfare. 
We consider problem instances with general downward-closed feasibility constraints $\feasible$ and provide three alternative clock auctions (depending on the amount of information available to the seller) that all achieve a $O(\log\log k)$ approximation guarantee, where $k$ is the number of maximal sets in $\feasible$. The more information the seller has, the simpler the proposed auction, providing sellers with three options that exhibit an interesting trade-off between access to information and auction simplicity.
These results are in stark contrast to the fact that prior-free deterministic clock auctions cannot achieve any bounded approximation even for very simple instances where $k=2$~\cite{DGR2017}.

\vspace{0.1in}{\bf Result 1: Deterministic single-price clock auction under full access to priors}  (Section~\ref{sec:single-price-dc}).
In the standard Bayesian setting, where buyer values are drawn from known distributions, we show that the $O(\log\log k)$ approximation guarantee can be provided by a deterministic clock auction that has a particularly simple structure. It is a member of a class of auctions, which we call {\em \caucts}, which use only a small portion of the power that clock auctions possess and, as we discuss in Section~\ref{sec:single-price-dc}, are closely related to the well-studied class of posted-price mechanisms.
Like posted price mechanisms, \caucts\ offer just a single price to each buyer. However, single-price clock auctions can defer the decision regarding which subset of buyers to serve until every buyer has responded to the price they were offered, whereas posted-price mechanisms need to serve every buyer who accepts the offered price.  
Our results show that this advantage allows us to achieve an exponential improvement over posted-price mechanisms, which cannot achieve better than a $O\left(\log{k}/\log{\log{k}}\right)$ approximation, even for a simple feasibility constraint where all the $k$ maximal feasible sets are disjoint \cite{BIK2007, R2016}. 

\vspace{0.1in}{\bf Result 2: Deterministic clock auction under restricted access to priors} (Section~\ref{sec:dc}).
Our second result is a deterministic auction which achieves the same approximation guarantee using only limited information regarding the distributions: it has access to the expected value of each bidder and to the expected value of the optimal solution, but has no additional information regarding the moments of $\disti$.
This is in line with recent work focusing on mechanism design with limited information~(e.g., \cite{AMDW2013,AM13,ColeR14,HuangMR18,correa2019prophet,correa2020two,ezra2018prophets,RWW2020,AKW19}). 
In contrast to our first auction, which leverages its unrestricted access to the distributions to offer only a single price to each bidder, our second auction may need to interact multiple times with each bidder, over a sequence of rounds. 
Nevertheless, its structure remains rather simple: every bidder faces the \emph{same} price $p$, which is gradually increased until either the set $A$ of active bidders is feasible, i.e., $A\in \feasible$, or some feasible subset of the active bidders $F\subset A$ reaches an appropriate revenue target.

\vspace{0.1in}{\bf Result 3: Randomized prior-free clock auction} (Section~\ref{sec:rand-clock}). 
Our third result applies to settings without any access to prior information.
Using the tools and intuition from the analysis of the first two auctions, we show that the ``hedging auction" introduced in \cite{CGS2022} gives an $O(\log\log k)$ approximation as well. 
This provides an exponential improvement over the $O(\sqrt{\log k})$ bound shown in \cite{CGS2022}.
The hedging auction is randomized and significantly more complicated than the deterministic auctions above. It randomly chooses between two alternatives. 
The first alternative uses non-uniform pricing (i.e., the price offered to each buyer in each round can be different from what is offered to others) and the price trajectory is determined by a complex price update process. The second alternative randomly ``samples" a subset of the buyers and then uses their values to decide which of the non-sampled buyers to serve. Although neither one of these alternatives achieves a good performance on its own, we show that the best of the two is guaranteed to perform well.
The hedging auction depends on randomization in two ways: i) it hedges between the two alternatives, and ii) it uses randomized sampling. 

\vspace{0.1in}

In Section~\ref{sec:lowerbound}, we complement our positive results with an inapproximability result that applies to all three settings we study.
Our lower bound holds for a very simple instance with two maximal feasible sets and applies even to Bayesian settings with independent values. This is in contrast to a previous lower bound which holds only under correlated priors \cite{CGS2022}.

\vspace{0.1in}
{\bf A hierarchy of clock auctions} (Section~\ref{sec:hierarchy}).
We conclude with a proposed ``hierarchy'' of clock auctions, depending on the amount of power that they use. This hierarchy sets the stage for an exciting new line of research with a broad range of open problems that would help us develop a more refined understanding of the guarantees that are achievable by a spectrum of gradually more powerful clock auctions. As an initial
step toward that direction, we provide some 
results and observations,
and discuss interesting connections between the construction of lower bounds for \caucts\ and the literature on large deviations in probability theory.

\subsection{The Compelling Advantages of (Deferred-Acceptance) Clock Auctions}\label{sec:properties}

The recent surge of interest in clock auctions, starting with the work of \citet{MS2014,MS2019}, has largely been motivated by the long list of appealing properties that these auctions satisfy, which make them a better fit for many real-world applications than the widely-studied sealed-bid auctions. In a sealed-bid auction, the buyers are asked to directly report their values to the auctioneer through a bid, and the auction then takes all this information into consideration in deciding who should be served and what the payments should be. The mechanism design literature has produced several sealed-bid auctions that are strategyproof, but experimental evidence suggests that people often misreport their values even when faced with simple strategyproof sealed-bid auctions, whereas they do follow their dominant strategy in clock auctions~\cite{KHL1987}. 

\textbf{Obvious strategyproofness.}
To provide theoretical justification for this phenomenon,~\citet{L2017} introduced the much more demanding property of {\em obvious strategyproofness} (OSP). In simple terms, an auction is OSP if it provides every bidder with an {\em obviously dominant} strategy: at any point during the auction, each bidder's utility in the best-case scenario if she deviates from this strategy is no more than her utility in the worst-case scenario if she follows this strategy. Using both experimental evidence and theoretical arguments, Li argued that even non-experts can recognize that an obviously dominant strategy is their optimal choice. He also showed that in clock auctions it is an obviously dominant strategy for every bidder to remain in the auction until the price she is offered exceeds her value. In fact, for single-parameter domains like the ones we study here,  Li showed that clock auctions are essentially the \emph{only} class of auctions that can achieve the stronger incentive property of OSP. Furthermore, an auction that satisfies OSP, like clock auctions do, is also guaranteed to be {\em weakly group-strategyproof}, i.e., even if the buyers got together and manipulated the auction in a coordinated way, they would not all benefit from that deviation. This is an incentive property that very few sealed-bid auctions satisfy. 
Apart from their improved incentives, clock auctions also satisfy properties related to transparency, privacy, and simplicity.

\textbf{Transparency.} Another shortcoming of sealed-bid auctions is that the bidders need to trust that the allocation and prices that the auction outputs are, in fact, the result of the precise computations dictated by the auction's rules. For example, they need to trust that the auctioneer will not charge them more, even though they have revealed exactly how much they are willing to pay, through their bid. Also, even if the bidders believe that the auctioneer is not malicious, they need to trust that there were no unintended errors in computing the outcome or the prices, as even very small errors could violate the incentive properties. This is in stark contrast to clock auctions where every bidder sees the prices in each round and need not even worry about how they were computed, as this does
not affect their incentives in any way.
In fact, clock auctions are \emph{credible}, an important property introduced by \citet{AL2020}, meaning that it is optimal for the auctioneer to follow the stated rules of the auction. When selling a single item \citet{AL2020} showed that an ascending clock auction is the \emph{unique} optimal auction that is both credible and strategyproof.

\textbf{Unconditional Winner Privacy.} Sealed-bid auctions also require every bidder to reveal to the auctioneer exactly how much they are willing to pay for the service. In clock auctions, bidders reveal their true value only if they choose to drop out, so the winners do not need to ever reveal how much they were actually willing to pay for the service (which is often very sensitive information). \citet{MS2019} termed this property \emph{unconditional winner privacy}, based on a notion 
defined earlier by \citet{BS05}, and demonstrate that the set of auctions which preserve unconditional winner privacy corresponds to the set of clock auctions.

\textbf{Simplicity without Sacrificing Sophistication.} A particularly exciting property of clock auctions from an algorithm designer's standpoint is that they achieve simplicity without sacrificing sophistication. In the line of work motivated by simplicity in mechanism design (e.g., \cite{HR2009,BILW2014,R2016Simple}), simplicity often implies that the algorithmic aspect of the auction is rather straightforward so that the participants can understand it. However, clock auctions provide a very simple interface to the buyers (the ascending prices), and it is always an obviously dominant strategy for them to remain in the auction while their price does not exceed their value, \emph{no matter how the prices are chosen}. Thus, the designer can still make elaborate algorithmic decisions in the ``back-end'', in choosing the sequence of prices that it offers, aiming 
to achieve good performance guarantees, while maintaining the incentives of these buyers, and a simple ``front-end'' that the buyers interact with.

\subsection{Other Related Work}
\label{sec:related}
In one of the first papers on the performance of clock auctions, ~\citet{DGR2017} proved a rather pessimistic result: there exists a seemingly simple family of instances with $n$ bidders, where no deterministic prior-free clock auction can achieve an approximation of $O(\log^{1-\epsilon} n)$ for any constant $\epsilon>0$. In these instances, the auction needs to decide between serving a single bidder or any subset of the remaining $n-1$ bidders, so it has just $k=2$ disjoint maximal feasible sets. This simple feasibility constraint clearly exhibits the impact of the clock auctions' information limitations, leading to an inapproximability result that grows arbitrarily with the number of agents, even though $k$ is constant. In recent work, \citet{CGS2022} designed a deterministic prior-free clock auction that guarantees a $O(\log n)$ approximation for general downward-closed feasibility constraints,  showing that this is essentially the class of worst-case instances. To overcome the obstacle posed by these instances, they proposed a randomized clock auction that achieves a $O(\sqrt{\log k})$ approximation for downward-closed feasibility constraints, implying a $O(\sqrt{\log n})$ approximation for the interesting class of feasibility constraints with disjoint maximal feasible sets (generalizing the lower bound instances). Our results in this paper achieve an exponential improvement over both of these results.

Apart from the work focusing on the design of clock auctions, some recent work has focused more broadly on the design of auctions that satisfy just one of the clock auction properties, such as obvious strategyproofness (OSP) \cite{DKV2020,ferraioli2021two} or credibility \cite{daskalakis2020simple,ferreira2020credible,EFW22}. Most relevant to our results, \citet{ferraioli2021two} showed that the aforementioned setting of two disjoint maximal sets remains an obstacle for any prior-free deterministic OSP mechanism by demonstrating that no such mechanism can achieve better than a $\sqrt{\log{n}}$-approximation.  We thus show an exponential gap between clock auctions with priors or randomization and  prior-free deterministic OSP mechanisms for this setting.

We defer additional discussion on related work to the relevant sections; we discuss related work on posted-price mechanisms in Section~\ref{sec:single-price-dc} and on auctions with limited information in Section~\ref{sec:dc}.

\section{Preliminaries}
A set $N$ of $n$ buyers request some service and each buyer $i$ has a private value $\vali$ indicating how much they are willing to pay for it; the vector of all values is denoted by $\vect{\val}=(\vali)_{i\in N}$. 
A feasibility constraint $\feasible \subseteq 2^{N}$ contains the subsets of buyers that can be served simultaneously. Throughout the paper we focus on feasibility constraints that are downward-closed, i.e., if $F\in \feasible$ then $F'\in \feasible$ for every $F'\subseteq F$. That is, if the set of buyers $F$ can be served, then any subset of $F$ can also be served; a very mild assumption satisfied by virtually all interesting constraint structures. 

Two examples of downward-closed feasibility sets are \dsets\ and \ksystem. 
In \dsets, the buyers are partitioned into $k$ disjoint groups, $S_1,\ldots,S_k$, and a set $F$ of buyers is feasible if and only if it is a subset of one of these groups, i.e., $\feasible=\{F\subseteq N \mid \exists S_j \mbox{ s.t. } F \subseteq S_j\}$. In \ksystem, each buyer $i\in N$ has a demand of size $c_i$ and the auctioneer has a supply constraint restricting her to serve only subsets whose total demand does not exceed the supply, which is normalized to 1, i.e., $\feasible=\{F\subseteq N \mid \sum_{i\in F} c_i \leq 1\}$. 

Our goal is to design auctions that serve a feasible subset of buyers $F\in \feasible$, aiming to maximize the social welfare, $SW(F)=\sum_{i \in F}{\vali}$. 
In the Bayesian setting, the value $\vali$ of each bidder $i$ is independently drawn from a distribution $\disti$ and the product distribution over all bidders is denoted by $\vect{\dist}=\times_{i\in N} \disti$. The social welfare of a Bayesian auction is, therefore, a random variable and, for simplicity, we use $\Exlong[\vect{\val} \sim \vect{\dist}]{\textsc{auc}}$ to denote the expected social welfare of an auction \auc. 
We evaluate the performance of our Bayesian clock auctions using the demanding benchmark of the expected optimal social welfare, i.e., $\opt = \vbench$.
We say an auction $\auc$ approximates the optimal welfare within a factor $\alpha$ if $\opt\leq \alpha \cdot \auc$
for every distribution $\vect{\dist}$. Recall that a Bayesian auction can use information regarding the distribution $\vect{\dist}$. 
When clear in the context, we omit the subscript from the expectation.
In the prior-free setting, the values of the bidders are chosen adversarially and the expected social welfare of a randomized auction is over its own randomness. Our goal in both settings is to design auctions with small approximation factors.

A clock auction is a multi-round mechanism in which each bidder faces a personalized ``clock'' price, which weakly increases over time.  At the outset of the auction, all bidders are ``active". An initial vector of prices $\vect{\pricei[]}^1 = \{\price^1_i\}_{i=1}^{n}$ is posted to the bidders, and every bidder may choose to permanently exit the auction (this bidder is no longer active) or stay active at the current price. In each round $t$, the auctioneer posts a new price $\price^t_i$ to each active bidders $i$, where $\price^t_i \geq \price^{t-1}_i$. Each active bidder can then choose to remain active or exit permanently. 
If in round $\tau$ the remaining active bidders form a feasible set $F\in\feasible$, the auction can terminate and serve each active bidder $i$ at price $\price^\tau_i$ (all bidders who have exited the auction are not served and pay $0$).
The prices in each round are chosen using any public information, such as the history of prices offered, the bidders' responses to these prices (whether they accepted them or not), and the feasibility structure. In a setting with prior information, the prior information can also be used to inform this choice.

\section{Deterministic Single-Price Clock Auction with Priors}
\label{sec:single-price-dc}
We begin by proposing a particularly simple class of clock auctions which we call \emph{single-price clock auctions}. Rather than asking each bidder to respond to multiple (increasing) price offers, these auctions offer a single price to each. As a result, they can  be implemented in a decentralized and asynchronous fashion: each bidder could arrive at a different time, respond to the price that they are offered, and depart. The mechanism first collects the responses of all bidders and then, using this information, decides which subset to serve. Single-price clock auctions take the following form:
\vspace{3pt}

\begin{algorithm}[H]
\SetAlgoLined
\LinesNumbered
\DontPrintSemicolon
 Calculate some vector of prices $\vect{\price}$ using public information\;
 Set the clock of each bidder $i$ to $\pricei$ (bidders then choose to remain active or exit)\;
 Among the bidders that chose to remain in the auction, select a ``preferred'' feasible subset $F$\;
 Reject all bidders $i \notin F$, accept all bidders in $i\in F$ and charge them $\pricei$\;
 \caption{Single-price clock auction template}
\end{algorithm}
\vspace{3pt}

In this section, we first discuss interesting connections between these auctions and the very well-studied class of posted-price mechanisms. We then show our main technical result of the section. Given a downward-closed feasibility constraint $\feasible$, let $\maxsets=(S_1,\ldots,S_k)$ denote the collection of maximal feasible sets in $\feasible$. We demonstrate that, using full access to the priors, we can design a single-price clock auction that achieves a $O(\log\log k)$-approximation to the benchmark welfare.  

\subsection{Connections to posted-price mechanisms}
Single-price clock auctions are closely related to the class of \emph{(sequential) posted-price mechanisms}, which has attracted a lot of attention in the algorithmic mechanism design literature (e.g.,~\cite{HKS2007, CHMS2009,KW2012,DK2015,DFKL2020,R2016, RS2017, FGL2014}). These mechanisms approach the buyers sequentially and offer each of them a take-it-or-leave-it price for being served (possibly depending on the history of responses). They are then served if and only if they accept the price offered to them.  These mechanisms are typically given access to the full value distributions that the buyer values are drawn from.  

This line of work draws upon the ``prophet inequality'' literature in optimal stopping theory~\cite{KS1977, KS1978}. The literature on sequential posted pricing (and prophet inequalities) includes different models, depending on the assumption regarding the arrival order of the buyers. The most common assumption is that buyers arrive in an adversarial order, but other models have been studied too. For example, in the so-called {\em ordered prophets} model, the mechanism designer can choose the order in which she approaches the buyers (see, e.g., \cite{CHMS2009, Y2011, AEEHKL2017, BGLPS2018}). Similarly, in the {\em prophet secretary} model, buyers are assumed to arrive in a random order (see, e.g., \cite{EHLM2017, EHKS2018, CSZ2019, AEEHKL2017, CFHOV2017}).\footnote{It is easy to verify that even ordered prophets can be simulated as multi-round clock auctions.}  

From an implementation standpoint, posted-price mechanisms are appealing because they are decentralized, like single-price clock auctions.  
However, single-price clock auctions do not commit to serving \emph{all} the buyers who accept the prices offered to them (they are not \emph{online} mechanisms). 
They instead defer their decision regarding which subset of buyers to serve  after all the buyers have responded to their offers, which provides these auctions with some additional information.

These mechanisms also differ in the level of adaptivity they may employ in setting prices (see, e.g., \cite{Luc2017,Cor2019}).  A non-adaptive mechanism commits to the prices before receiving any responses, whereas an adaptive mechanism can offer prices sequentially, conditioning the offers on the responses of previous buyers. 
A {\cauct} uses non-adaptive pricing (every buyer's price is determined in advance). 
Thus, comparing the performance of {\caucts} to that of posted-price mechanisms has interesting implications regarding the power of \emph{deferred-acceptance}, that the former have, relative to the power of \emph{price-adaptivity}, possessed by the latter.  

\subsection{Achieving the $O(\log{\log{k}})$ bound}
In what follows, we demonstrate that if the auctioneer has access to the entire product distribution $\vect{\dist} = \times_{i \in \bidders} \distone_i$ of bidder values, she can design a \cauct\ which obtains a $O(\log \log k)$ approximation guarantee.   This is in contrast to posted-price mechanisms which we know cannot achieve better than a $O\left(\frac{\log{k}}{\log{\log{k}}}\right)$ approximation even for \dsets\ \cite{BIK2007, R2016}. Thus, our single-price clock auction outperforms all posted-price mechanisms. In this light, our result suggests that, in our setting, the ability to defer the acceptance decisions can be significantly more useful than the ability to adapt the prices.

\begin{theorem}\label{th:k_sets_single}
	Every downward-closed setting with $k$ maximal sets admits a deterministic \cauct\ that obtains an $O(\log{\log k})$ approximation to the expected optimal welfare if the full distribution $\dist$ is known. 
\end{theorem}

Before deriving our single-price clock auction, we first decompose the benchmark welfare into two parts. For each realization of bidder values there is some maximal feasible set $S\in \maxsets$ that achieves the maximum social welfare. This set could have a lot of ``low'' value bidders, a few ``high'' value bidders, or a mix of the two. To prove that our auction approximates $\opt$, we can then divide $\opt$ into two parts: the portion of this expected welfare that comes from ``low'' value bidders and the portion that comes from ``high'' value bidders. We show that if the low-value contribution is greater, the approximation factor is achieved by offering all bidders a price of $0$ and by accepting only the bidders from the set with highest expectation. 
On the other hand, if the low-value contribution is not within $\log\log{k}$ factor of the benchmark, then the high-value part must make up nearly all of the expected optimal welfare. 
We show that the vast majority of the high-value contribution comes from sets with a logarithmic (or fewer) number of high value bidders.
To define the partition of bidders into high and low values, we set a threshold $t_S$ for each maximal feasible set $ S\in\maxsets$ so that the expected number of bidders above the threshold is $\log k$. That is, let 
\[
S(t,\vals)\eqdef \{i\in S \vert \vali > t \}
\]
Then, the threshold $t_S$ is the value satisfying\footnote{For simplicity of presentation we assume that such a value exists (e.g., in cases where the expected value is continuous). Otherwise, we set the threshold $t_S$ to be the value satisfying $\inf_{t_S} \Ex[\vals]{\setsize{S(t_S,\vals)}} \leq \log{k}$.  The same is true for the $t'$ value described in the proof of Lemma \ref{lem:chernoff1}.}
\[
\Ex[\vals]{\setsize{S(t_S,\vals)}} = \log{k}.
\]
The choice of $\log k$ is a critical point at which it is still likely that for all $k$ maximal sets, the number of bidders exceeding their corresponding set's threshold $t_S$ is within a constant factor of the expectation. 
This is cast in the following lemma.

\begin{lemma}\label{lem:chernoff1}
	Let $t$ be a threshold for a set $S$ such that $\Ex[\vals]{\setsize{S(t,\vals)}} = \log{k}$. Then 
	\[
	\prob[\vals]{\exists x \in [0, t):~ \setsize{S(x, \vals)} > 10\cdot \Ex[\vals]{\setsize{S(x,\vals)}}} = o(1/k^2).
	\]
\end{lemma}

\begin{proof}
	First consider the threshold $t$ such that $\Ex[\vals]{\setsize{S(t,\vals)}} = \log k$ and an interval of threshold values 
	$[t', t)$ such that $\Ex[\vals]{\setsize{S(t',\vals)}} = 2\cdot\Ex[\vals]{\setsize{S(t,\vals)}}$.  We would like to have 
	$\setsize{S(x,\vals)} < 10 \Ex[\vals]{\setsize{S(x,\vals)}}$ for all $x \in [t', t)$. Consider a set of independent $\{0,1\}$-value random variables $\{\xi_i\}_{i\in S}$ where $\xi_i\eqdef \ind{v_i\ge t'}$. Note that $\sum_{i\in S}\xi_i=\setsize{S(t',\vals)}$. 
	By Chernoff bound, it follows that 
	\[
	\prob[\vals]{\vphantom{\sum}\setsize{S(t',\vals)} > 5\Ex[\vals]{\setsize{S(t',\vals)}}} < \left(\frac{e^4}{5^5}\right)^{2\cdot\log k} 
	= o(1/k^2).
	\] 
	For all $x \in [t', t)$ we have $\setsize{S(x,\vals)} \leq \setsize{S(t',\vals)}$ and $\Ex[\vals]{\setsize{S(x,\vals)}} \geq \Ex[\vals]{\setsize{S(t,\vals)}} = \frac{1}{2}\cdot\Ex[\vals]{\setsize{S(t',\vals)}}$. Thus, if $\setsize{S(t',\vals)} < 5\Ex{\setsize{S(t',\vals)}}$, then $\setsize{S(x,\vals)} < 10\Ex{\setsize{S(x,\vals)}}$ for all $x \in [t',t)$.  Hence,
	\[
	\prob[\vals]{\vphantom{\sum}\exists x \in [t',t)~:~\setsize{S(x,\vals)} > 10\Ex[\vals]{\setsize{S(x,\vals)}}} = o(1/k^2).
	\] 
	Consider the list of intervals $[r_0, r_1),[r_1,r_2),...,[r_{\ell-1}, r_{\ell})$, where $\Ex[\vals]{\setsize{S(r_j,\vals)}} = 2\Ex[\vals]{\setsize{S(r_{j-1},\vals)}}$ for all $j=1,\ldots,\ell$ and such that $r_0 = 0$ (or alternatively the lowest value in the support of any bidder) and $r_{\ell} = t$.  We may apply the same Chernoff bound argument that we did for the interval $[t', t)$ to any of these intervals and get a bound of $o(1/k^{2^{\ell-j+1}})$ for the interval $[r_{j-1},r_{j})$. Then, by the union bound over all the intervals we get an upper bound of 
	\[
	\prob[\vals]{ \exists x \in \cup_{j=1}^{j=\ell} [r_{j-1}, r_j)~:~ \setsize{S(x, \vals)} > 10\cdot \Ex[\vals]{\setsize{S(x,\vals)}}}\le 
	\sum_{j=1}^{j=\ell}o\left(  \frac{1}{k^{2^{\ell-j+1}}} \right)=o\left(\frac{1}{k^2}\right).
	\]
	Consequently, $\prob[\vals]{\exists x \in [0, t) ~:~ \setsize{S(x,\vals)} > 10\Ex[\vals]{\setsize{S(x,\vals)}}} = o(1/k^2) $.
\end{proof}

We now write the following  low-high decomposition of the benchmark.

\be
\label{eq:low_high_decomposition}
\opt=\Ex{\max_{S\in\maxsets}\sum_{i\in S}\vali}\le
\underbrace{\Ex{\max_{S\in\maxsets}\sum_{i\in S}\min(\vali,t_S)}}_{\text{$\mathsf{LOW}$}}+
\underbrace{
\Ex{\max_{S\in\maxsets}\sum_{i\in S}\vali\cdot\ind{\vali > t_S}}}_{\text{$\mathsf{HIGH}$}},
\ee
where for every bidder $i$ 
we define two regimes with respect to every set $S$ such that $i \in S$:
\[
\text{low-value part: } \valciS\eqdef\min\{t_S,\vali\},
\quad\quad
\text{high-value part: } \valhiS\eqdef \vali\cdot\ind{\vali > t_S},
\]
so $\valciS\in[0,t_S]$, and if $\valhiS > 0$ then it must be that $\valhiS > t_S.$ 
Also, one can easily verify that 
$\vali\le \valciS+\valhiS$ for any bidder $i$ and any set $S$ such that $i \in S$.  

We first establish the following core-tail decomposition of $\mathsf{HIGH}= \Ex[\vals]{\max_{S\in\maxsets}\sum_{i\in S}\valhi}$ that  allows us to reduce the problem to the case where all sets $S\in\maxsets$ have size of at most $O(\log k)$. We define the core part of $\mathsf{HIGH}$ as
\[
\highcore\eqdef\Ex[\vals]{\max_{S\in\maxsets}\InParentheses{\sum_{i\in S}\valhi \cdot\ind{|S(t_S,\vals)|\le 10\log k +1}}}.
\]
It follows that
\begin{eqnarray}
\label{eq:high_core_tail_decomposition}
 \mathsf{HIGH} &\le & \highcore
+\Ex[\vals]{\max_{S\in\maxsets}\InParentheses{\sum_{i\in S}\valhi \cdot\ind{|S(t_S,\vals)| > 10\log k +1}}}\nonumber \\
&\le & \highcore +\sum_{S\in\maxsets}\Ex[\vals]{\sum_{i\in S}\valhi \cdot\ind{|S(t_S,\vals)| > 10\log k +1}}\nonumber \\
&\le & \highcore +\sum_{S\in\maxsets}\Ex[\vals]{\sum_{i\in S}\vali \cdot\ind{|S(t_S,\vals)-\{i\}|> 10\log k}}\nonumber \\
&= & 
\highcore +
\underbrace{\sum_{S\in\maxsets}\sum_{i\in S}\Ex{\vali}\cdot\prob[\vals]{|S(t_S,\vals)-\{i\}|> 10\log k}}_{\text{$\hightail$}}.
\end{eqnarray}

In the second inequality, we replace max by sum. In the third, we modify the event in the right hand side so that it is independent of the value of bidder $i$; we also use the fact that $\valhi \leq v_i$. In the last inequality we use the independence in the right hand side to write the expectation as a product.

We claim that $\hightail$ can be covered by the simple single-price clock auction which sets zero prices to all bidders and selects the set with maximum expected value. 
Note that it is extremely unlikely that there are more than $10\log k$ bidders above the threshold $t_S$ in any set $S$, i.e., the event that $|S(t_S,\vals)|>10\Ex{S(t_S,\vals)}+1=10\log k+1$ for each $S\in\maxsets$ has probability of $o(1/k^2)$.  Indeed, Lemma~\ref{lem:chernoff1} gives an even stronger tail probability bound. 
Note that a simple single-price clock auction with prices at zero would select the feasible set
with maximum expected value and get expected social welfare of  $\max_{S\in\maxsets}\Ex{\sum_{i\in S}\vali}\ge\frac{1}{k}\sum_{S\in\maxsets}\sum_{i\in S}\Ex{\vali}$, which is more than $\sum_{S\in\maxsets}\sum_{i\in S}\Ex{\vali}\cdot o(1/k^2)$. 
Recall that 
\[\mathsf{LOW}=\Ex{\max_{S\in\maxsets}\sum_{i\in S}\valciS}=\Ex{\max_{S\in\maxsets}\sum_{i\in S}\min(\vali,t_S)}.\]
Next we show that the single-price clock auction with zero prices is a $O(1)$-approximation to the $\mathsf{LOW}$ term in~\eqref{eq:low_high_decomposition}.  To this end, it is useful to rewrite $\sum_{i\in S}\valciS=\int_0^{t_S}|S(x,\vals)| \dd x$. By linearity of expectations we also have 
\[
\Ex[\vals]{\sum_{i\in S}\valciS}=\Ex[\vals]{\int_0^{t_S}|S(x,\vals)| \dd x}=
\int_0^{t_S}\Ex[\vals]{|S(x,\vals)|} \dd x
\] 
We use a slightly more complex core-tail decomposition than \eqref{eq:high_core_tail_decomposition} for the $\mathsf{LOW}$ term of \eqref{eq:low_high_decomposition}.

\[\lowcore\eqdef \Ex[\vals]{\max_{S\in\maxsets}\InParentheses{\sum_{i\in S}\valciS \cdot
\ind{\forall x\in[0,t_S)~:~ |S(x,\vals)|\le 10\Ex[\vals]{|S(x,\vals)|}+1}}}
\]
It follows that
\begin{align}
\label{eq:low_core_tail_decomposition}
\mathsf{LOW}
&\le \lowcore
+ \Ex[\vals]{\max_{S\in\maxsets}\InParentheses{\sum_{i\in S}\valciS \cdot
\ind{\exists x\in[0,t_S)~:~ |S(x,\vals)|> 10\Ex[\vals]{|S(x,\vals)|}+1}}} \nonumber \\  
&\le \lowcore + \sum_{S\in\maxsets} \Ex[\vals]{\sum_{i\in S}\valciS \cdot
\ind{\exists x\in[0,t_S)~:~ |S(x,\vals)|> 10\Ex[\vals]{|S(x,\vals)|}+1}} \nonumber \\  
&\le \lowcore + \sum_{S\in\maxsets} \Ex[\vals]{\sum_{i\in S}\vali \cdot
\ind{\exists x\in[0,t_S)~:~ |S(x,\vals)-\{i\}|> 10\Ex[\vals]{|S(x,\vals)| }}} \nonumber \\  
&= \lowcore +
\underbrace{\sum_{S\in\maxsets}\sum_{i\in S}\Ex{\vali}\cdot\prob[\vals]{\exists x\in[0,t_S)~:~|S(x,\vals)-\{i\}|> 10\Ex[\vals]{|S(x,\vals)|}}}_{\text{$\lowtail$}}
\end{align}

We again may cover the tail (second) term on the right hand side of \eqref{eq:low_core_tail_decomposition} since by Lemma~\ref{lem:chernoff1} $\prob[\vals]{\exists x\in[0,t_S):~|S(x,\vals)-\{i\}|> 10\Ex[\valcs]{|S(x,\vals)|}} = o(1/k^2)$. 
Hence, $\max_{S\in\maxsets}\Ex{\sum_{i\in S}\vali}\ge\frac{1}{k}\sum_{S\in\maxsets}\sum_{i\in S}\vali$ is at least $\Theta(k)$ times the value of the $\lowtail$ term in the right hand side of~\eqref{eq:low_core_tail_decomposition}.
Let us now estimate the $\lowcore$ term in \eqref{eq:low_core_tail_decomposition}. Let 
\[
\event(\vals)\eqdef \ind{\forall x\in[0,t_S)~:~ |S(x,\vals)|\le 10\Ex[\vals]{|S(x,\vals)|}+1}
\] 
be the indicator of the event that for all thresholds $x$ in the range from $0$ to $t_S$ the number of bidders in $S$ with values above $x$ does not exceed the expected number by a factor of $10$. Then 
\begin{multline*}
\lowcore=\Ex[\vals]{\max_{S\in\maxsets}\InParentheses{\sum_{i\in S}\valciS\cdot \event(\vals)}}=
\Ex[\vals]{\max_{S\in\maxsets}\InParentheses{\event(\vals)\cdot\int_0^{t_S}|S(x,\vals)| \dd x }}\\
\le \Ex[\vals]{\max_{S\in\maxsets}\int_0^{t_S}(10\Ex[\vals]{|S(x,\vals)|}+1) \dd x }=
\max_{S\in\maxsets}\int_0^{t_S}(10\Ex[\vals]{|S(x,\vals)|}+1) \dd x\\
\leq \max_{S\in\maxsets}\int_0^{t_S}(11\Ex[\vals]{|S(x,\vals)|}) \dd x
=11\max_{S\in\maxsets}\Ex[\vals]{\int_0^{t_S}|S(x,\vals)|\dd x}
=11\max_{S\in\maxsets}\Ex[\vals]{\sum_{i\in S}\vali}.
\end{multline*}
In other words, we get that the single-price clock auction with zero prices gets a constant approximation to the $\mathsf{LOW}$ term of \eqref{eq:low_high_decomposition}.

Inequalities \eqref{eq:low_high_decomposition} and \eqref{eq:high_core_tail_decomposition} give us an upper bound $\opt\le \mathsf{LOW}+ \hightail+\highcore$, where the welfare of a single-price clock auction with zero prices covers the terms $\mathsf{LOW}$ and $\hightail$.

It remains to show we can obtain a $O(\log \log k)$ approximation to the $\highcore$ term of the benchmark using a single-price clock auction.  Our next lemma gives an upper bound on the $\highcore$ term.  In $\highcore$, we may effectively assume that each set $S \in \maxsets$ is of size at most $10\log k + 1$.

\begin{lemma}\label{lem:singlefixedprice} 
Let $m$ be any positive integer, and let 
\[
\Delta\eqdef\Ex{\max_{S\in\maxsets}\InParentheses{\sum_{i\in S}\valhi \cdot\ind{|S(t_S,\vals)|\le m}}}.
\]
There exists a uniform price $\price$ such that 
\[
\Delta
\le 
O(\log m)\cdot \Ex{\max_{S\in\maxsets}\sum_{i\in S}\Ex{\vali ~\big|~ \vali\ge\price}\cdot\ind{\vali\ge\price}}.
\]
\end{lemma}
Note that $\Ex{\max_{S\in\maxsets}\sum_{i\in S}\Ex{\vali ~\big|~ \vali\ge\price}\cdot\ind{\vali\ge\price}}$ is precisely the expected welfare we get in a single-price clock auction with a uniform price $p$.
Lemma~\ref{lem:singlefixedprice} shows that the best choice of a single uniform price 
will give an $O(\log m)$-approximation to $\Delta$, which translates to $O(\log\log k)$-approximation to $\highcore$\footnote{This is not very surprising given similar statements in the literature, see, e.g., Corollary 1 in \cite{R2016}.} with $m = 10 \log{k} + 1$. We give a proof sketch here; the complete proof appears in  Appendix~\ref{sec:addlproofs}. 

\begin{proof}[Proof sketch for Lemma~\ref{lem:singlefixedprice}]
Consider the contribution of bidders with value less than $2\cdot \Delta$ to the benchmark (equals to $\Delta$). The $O(\log m)$ approximation guarantee is achieved by the revenue (not the welfare) of a \cauct.  Observe that the contribution to the benchmark of bidders with $\valhi\le\frac{\Delta}{2m}$ is at most $\frac{\Delta}{2}$ so by ignoring all such bidders we lose at most a constant factor of $\Delta$.  Consider partitioning all remaining bidders in $N$ into $\log{m} + 2$  buckets $\left[\frac{\Delta}{2m}, \frac{\Delta}{m}\right], \left[\frac{\Delta}{m}, \frac{2\Delta}{m}\right], \cdots, \left[\Delta, 2\Delta\right]$ based on their value. 
Let $B$ be the bucket with the largest expected contribution among all these $O(\log m)$ buckets.
The revenue obtained by posting the lower bound of the interval for $B$ is a $2$-approximation to the welfare contained in $B$ for each set and, in particular, the one with maximum realized value.

On the other hand, the uniform price auction with $p=2\Delta$ covers the constant fraction of the contribution from the bidders with $\valhi\ge 2\Delta$. Indeed, if $\vali\ge 2\Delta$ for some $i$, then with a constant probability $i$ is the only such bidder. I.e., a \cauct\ with $p=2\Delta$ takes all such bidders with constant probability and achieves $O(1)$ approximation to the contribution of bidders with $\valhi\ge 2\Delta$.
\end{proof}

Thus, a single-price clock auction achieves a $O(\log \log k)$ approximation to the $\mathsf{HIGH}$ term in the right side of \eqref{eq:low_high_decomposition}. Combined with the constant approximation given by a single-price clock auction with zero prices to the $\mathsf{LOW}$ term, we are ready to present the proof of Theorem \ref{th:k_sets_single}.

\begin{proof}[Proof of Theorem \ref{th:k_sets_single}]
By the analysis above, we can achieve the desired approximation using one of the following uniform price clock auctions (the one with the highest expected social welfare):
	\begin{enumerate}
		\item Either let $\pricei=0$ for all $i\in [n]$,
		\item or choose a $j\in[0, \log (10\log k+1)]$ and let $ \pricei=\Delta\cdot 2^{1-j}$ for all $i\in [n]$, where
		\[
		\Delta=\Ex{\max_{S\in\maxsets}\InParentheses{\sum_{i\in S}\vali\cdot\ind{\vali>t_S} \cdot\ind{|S(t_S,\vals)|\le 10\log k+1}}}.
		\]
	\end{enumerate}
	Note that these prices can be found using only the thresholds $t_S$ and the value of $\Delta$.  Both $t_S$ and $\Delta$ can be inferred from the bidder distributions and constraint system, which are public information.
 The uniform price of zero covers the $\mathsf{LOW}$ term of our benchmark decomposition \eqref{eq:low_high_decomposition} and the $\hightail$ portion of \eqref{eq:high_core_tail_decomposition}, whereas the uniform prices of $\Delta \cdot 2^{1-j}$ cover the $\highcore$ portion of \eqref{eq:high_core_tail_decomposition} (which, together, cover the $\mathsf{HIGH}$ term of \eqref{eq:low_high_decomposition}).
\end{proof}

\subsection{Alternative Parameterizations}
\label{sec:alternative}

Similar to~\cite{CGS2022}, our bounds are parameterized by $k$, the number of maximal feasible sets of the instance. Alternative parameterizations that have been considered include, e.g., the number of agents, $n$, or the size of the largest feasible set, $r$. For the class of instances with \dsets, which are known to pose an obstacle for both posted-price mechanisms and deterministic prior-free clock auctions, our bounds improve prior results with respect to both parameters.

When using the number of agents as a parameter, 
it is known that even for \dsets, no 
posted-price mechanism can achieve better than a $O\left(\log{n}/\log{\log{n}}\right)$ approximation \cite{BIK2007, R2016}, and no deterministic prior-free clock auction can achieve a $O(\log^{1-\epsilon}{n})$ approximation for any constant $\epsilon>0$ \cite{DGR2017}\footnote{This result and all the results in Sections \ref{sec:single-price-dc} and \ref{sec:dc} focus on deterministic mechanisms.}.
Since, for \dsets\ the number of maximal feasible sets is $k\leq n$, our auctions achieve a $O(\log \log n)$ approximation, which is an exponential improvement compared to both posted-price mechanisms and deterministic prior-free clock auctions.

When using the size of the largest feasible set, $r$, as the parameter, \citet{R2016} demonstrated that posted-price mechanisms can achieve a $O(\log{n} \log{r})$ approximation for general downward-closed feasibility constraints. Applying the mechanism described in Theorem~\ref{th:k_sets_single}, but with $\Delta = \opt$ and $m = r$ (where $m$ is from the statement of Lemma~\ref{lem:singlefixedprice}), we get an improved bound with respect to this parameter as well.
Indeed, Lemma~\ref{lem:singlefixedprice} holds for any threshold $t_S$, in particular for $t_S=0$. Since the size of every feasible set is at most $r$, we have $\ind{|S(t_S,\vals)|\le r}=1$ for all $\vals$. Lemma~\ref{lem:singlefixedprice} then shows the existence of a single price $p$ such that the corresponding \cauct\ achieves a $O(\log r)$ approximation to $\Delta = \opt$.

\section{Deterministic Clock Auction with Limited Information}
\label{sec:dc}
Our main result in this section is a deterministic uniform price clock auction that achieves an approximation of $O(\log\log k)$ for any downward-closed $\feasible$, using only limited prior information regarding the distributions: the expected value of each bidder and the expected value of the optimal solution $\opt=\Ex[\vals]{\max_{S\in \maxsets}\sum_{i\in S}\vali}$. This is in stark contrast to deterministic clock auctions without this information, which cannot achieve any constant approximation even for $k=2$~\cite{DGR2017}.  In other words, by employing a slightly more complicated auction procedure, i.e., a uniform ascending price procedure, one can circumvent the need for full distributional knowledge by ``discovering'' the appropriate price from Section \ref{sec:single-price-dc} which gave the approximation guarantee.

This clock auction contributes to the growing literature aiming to design more robust auctions that do not require full access to the priors. For example, in a ``parametric auction''~\cite{AMDW2013,AM13} the seller only knows some parameters of the distributions (e.g., the mean). Another related literature is the one focusing on the design of auctions that instead have access to a limited number of samples from each distribution (e.g.,~\cite{ColeR14,HuangMR18,correa2019prophet,correa2020two,ezra2018prophets,RWW2020,AKW19}).

Our proposed auction starts by checking whether there exists some maximal feasible set $S\in \maxsets$ whose expected value is at least $\opt/\log{\log{k}}$. If such a set exists, then the desired approximation can easily be achieved by accepting all the bidders of that set, for a price of 0, and rejecting all other bidders (we refer to this part of the auction as ``$\auco$''). In the more interesting case where no such set exists, our auction uniformly raises the price of every active bidder using some arbitrarily small step $\delta>0$, giving each bidder the opportunity to drop out after each price increase. While we present the auction as offering this price simultaneously to all bidders for simplicity, we instead offer each bidder the increased price one at a time in a consistent, but arbitrary, order. The auction terminates when the remaining active bidders become feasible, or the revenue of some feasible subset of active bidders reaches a desired goal of $g=\opt/(4\alpha)$ where $\alpha$ is the approximation guarantee given by Corollary~\ref{lem:fixedprice} in Appendix \ref{sec:addlproofs} (we refer to this part of the auction as ``$\uprice$'').  Given its similarity to the proof of Theorem \ref{th:k_sets_single}, we defer the full proof of Theorem \ref{th:k_sets} to Appendix \ref{sec:addlproofs} and instead provide a sketch below, highlighting the key differences.
\vspace{3pt}

\LinesNumbered
\begin{algorithm}[H]
	\SetAlgoLined
	\DontPrintSemicolon
	Let all bidders be active $\activebidders=N$ and set clock price $p=0$\;  
	\If{$\max_{S\in\maxsets}\sum_{i\in S}\Ex{\vali}\geq \opt/\log{\log{k}}$} {
		Accept the bidders in $\arg\max_{S\in\maxsets}\sum_{i\in S}\Ex{\vali}$, charge each of them $p=0$, and reject everyone else
	}
	\Else{
		\textbf{Repeat:} increase $p$ by some $\delta>0$ and offer $p$ to all $i\in \activebidders$; remove from $\activebidders$ bidders who reject $p$\;
		\textbf{Until:} $\activebidders$ is feasible (i.e., $\activebidders\in \feasible$) \textbf{or} there exists some $F\subset \activebidders$ such that $F\in \feasible$ and $|F|\cdot p\ge g$\; 
		Accept the largest feasible set $F\subseteq \activebidders$, charge each $i\in F$ a price of $p$, and reject everyone else
	}
	\caption{A uniform-price clock auction with limited information}
	\label{mech:uniform}
\end{algorithm}
\vspace{3pt}

\begin{theorem}
	\label{th:k_sets}
	Mechanism~\ref{mech:uniform} achieves a $O(\log{\log k})$ approximation to the expected optimal welfare for any downward-closed $\feasible$ using only $\Ex{\vali}$ for each $i\in N$ and $\opt=\Ex[\vals]{\max_{S\in \maxsets}\sum_{i\in S}\vali}$. 
\end{theorem}

\begin{proof}[Proof Sketch.]

Our proof proceeds similarly as in the last section, again employing our benchmark decomposition, inequality \eqref{eq:low_high_decomposition}.  Note that $\auco$ is exactly the single-price clock auction with all prices equal to zero, which, from our previous section, we know covers the $\lowcore$, $\lowtail$, and $\hightail$ portions of our benchmark.  
We can then assume that $\highcore$ is a $2$ approximation to $\opt$, as otherwise $\auco$ would be a $O(\log{\log{k}})$ approximation to $\opt$. Thus it suffices to show that $\uprice$ gives a $O(\log\log k)$ approximation to $\highcore$. 

From the proof of Lemma \ref{lem:singlefixedprice} in Appendix \ref{sec:addlproofs} we may conclude the following related corollary.
\begin{corollary}\label{lem:fixedprice} 
Let $m$ be any positive integer, and let 
\[
\Delta\eqdef\Ex{\max_{S\in\maxsets}\InParentheses{\sum_{i\in S}\valhi \cdot\ind{|S(t_S,\vals)|\le m}}}.
\]
There exist an $\alpha=O(\log m)$ and a uniform price $\price$ such that 
\[
\Delta\le\alpha\cdot\price\cdot\Ex{\max_{S\in\maxsets}\sum_{i\in S}\ind{\vali\ge\price}}
\quad\text{or}\quad
\Delta\le\alpha\cdot\Ex{\max_{i\in N}\Ex{\vali ~\big|~ \vali\ge\price}\cdot\ind{\vali\ge\price}}.
\]
\end{corollary}
Corollary \ref{lem:fixedprice}, thus, proves the existence of a uniform price $p$ such that either the expected revenue obtained by serving the set with the largest number of bidders with value above $p$ is a $O(\log \log k)$ approximation to $\highcore$ or the expected welfare from serving the single highest value bidder is a $O(\log \log k)$ approximation to $\highcore$. Without access to the full prior distributions, however, we cannot calculate this price in advance.  Instead, we use $\uprice$ to gradually sweep through a full range of prices, aiming to find one that yields revenue at least $g = \frac{\opt}{4\alpha}$. However, this could fail if $\uprice$ reaches a lower price that also satisfies the revenue target and stops too soon, selecting a set with low social welfare.  We address this by a more careful analysis of the expected welfare of $\uprice$, conditioning on the expected size $\Ex{|T(\vals)|}$ of the largest feasible set $T(\vals)$ of bidders who would accept price $p$.  When $\Ex{|T(\vals)|} \geq 8$ we show that with constant probability we serve at least half the expected number of bidders and obtain the $O(\log \log k)$ approximation from the \emph{revenue} of $\uprice$.  On the other hand, when $\Ex{|T(\vals)|} < 8$ we show that the expected \emph{social welfare} of $\uprice$ is within a constant factor of the single highest value bidder, which then gives us the desired approximation.
\end{proof}

\section{Randomized Prior-Free Clock Auction}\label{sec:rand-clock}

\newcommand{\thresh}[1]{\tau(#1)}
\newcommand{\topS}[1]{#1_{\text{top}}}
\newcommand{\threshSubset}[2]{#1(#2)}
\newcommand{\sampled}[1]{T(#1)}
\newcommand{\unsampled}[1]{U(#1)}

\newcommand{\maximalSets}{\mathcal{S}}

\newcommand{\hybridauc}{Hedging Auction}

\newcommand{\selectionThreshold}{\hat{q}}

In this section, we consider prior-free settings; i.e., settings where no information is given about the bidder values.
Using the analytical tools and intuition developed in the previous sections, we can now provide a tighter analysis of the randomized clock auction (the ``hedging auction'') introduced in \cite{CGS2022}, showing that it achieves a $O(\log \log k)$ approximation (an exponential improvement over the $O(\sqrt{\log k})$ bound shown in \cite{CGS2022}). 
This auction provides guarantees even in prior-free settings, but this is at the expense of simplicity and practicality; namely, it is randomized and uses non-uniform pricing.
Note that our randomized clock auctions are distributions over deterministic clock auctions, and thus they satisfy all the properties of deterministic clock auctions \emph{ex post} (rather than in expectation).

For completeness, we first explain how this auction works and provide some intuition regarding its approximation guarantee. 
The hedging function runs either the ``water-filling clock auction'' (WFCA) or the ``sampling auction", each with probability $1/2$.
We begin by defining the ``water-filling clock auction'' (WFCA).
Let $P=\{p= x \cdot \epsilon : x\in \mathbb{N}\}$ denote the set of possible prices that the WFCA considers.  As was demonstrated in \cite{CGS2022}, the WFCA obtains welfare at least $r^*/2$ where $r^* = \text{max}_{p \in P, F \in \feasible}\left\{p \cdot |i \in F : \vali \geq p|\right\}$, i.e., the maximum possible revenue an auction could obtain by offering one of the prices in $P$ to all bidders.  This then translates to an $O(\log{|S|})$-approximation for any given set $S$.

	\vspace{3pt}

	\begin{algorithm}[H]
		\SetKwInOut{Input}{Input}
		
		let $t\leftarrow 0$, $\activebidders \leftarrow \bidders$, and $p_i \leftarrow 0$ for all $i \in \bidders$
		
		\While{$\activebidders \notin \feasible$}{
			$t\leftarrow t+1$
			
			let $W \leftarrow \arg\max_{S \in \feasible: S \subseteq \activebidders}\left\{\sum_{i \in S}{\pricei}\right\}$ be the latest set of conditional winners \label{alg:line:W}

			let $\ell \leftarrow \min_{i \in \activebidders \setminus W}\{\pricei\}$ be the lowest price among conditional losers \label{alg:line:p}

			\ForEach{bidder $i \in \activebidders \setminus W$ with $\pricei=\ell$}{
				update $\pricei \leftarrow \pricei + \epsilon$
				
				\If{$i$ rejects updated price}
				{
					let $\activebidders \leftarrow \activebidders \setminus \{ i \}$
				}
			}
		}
		\Return{$\activebidders$} 
		\caption{The deterministic \emph{water-filling clock auction} (WFCA)}
		\label{alg:wfca}
	\end{algorithm}
	\vspace{3pt}

	We next describe the sampling auction. The sampling auction ``samples'' bidders independently with probability $1/2$ by raising their clock until they reject.  We let $T \subseteq N$ denote the sampled bidders and $U = N \setminus T$ denote the unsampled bidders.  This sampling process ideally allows us to estimate the value of each $S$ since $v(S) = 2 \cdot \Ex{v(S \cap U)} = 2 \cdot \Ex{v(S \cap T)}$.
	\vspace{3pt}
	
	\begin{algorithm}[H]
		\SetKwInOut{Input}{Input}
		\Input{The set $\maximalSets$ of all maximal sets in $\feasible$}
		
		let $T \leftarrow \emptyset$
		
		\For{each bidder $i \in \bidders$}{
			With probability $1/2$: increase the clock of $i$ until she rejects and let $T \leftarrow T \cup \{i\}$
			
		}
		
		Let $R\gets \argmax_{S\in \maximalSets}v(S\cap T)$

		\Return{$R \setminus T$} 
		\caption{The Sampling Auction:  A randomized clock auction}
		\label{alg:sampling-auction}
	\end{algorithm}
	\vspace{3pt}
	
	Intuitively, the water-filling auction and sampling auction work when sets of high value are small and large, respectively.  This is because when sets are small the optimal revenue in hindsight $r^*$ is close in value to the optimal welfare and when sets are large we expect with reasonably high probability that a random sample of the bidders will give a good estimate of the total value of the set.
    The ``Hedging Auction'' leverages this intuition.  It combines the benefits of the water-filling clock auction (WFCA) and the sampling auction by running one or the other with equal probability.  
    \vspace{3pt}
	
	\begin{algorithm}[H]
		\SetKwInOut{Input}{Input}
		\Input{The feasibility constraint $\feasible$, and set $\maximalSets$ of all maximal sets in $\feasible$}
		
		With probability $1/2$: Run the WFCA (Mechanism \ref{alg:wfca}) on input $\feasible$
		
		With the remaining probability $1/2$: Run the Sampling Auction (Mechanism \ref{alg:sampling-auction}) on input $\feasible$
		
		\caption{The \hybridauc\ }
		\label{alg:hedging}
	\end{algorithm}
	\vspace{3pt}
	
	\begin{theorem}\label{thm:hedging}
		The Hedging Auction obtains a $O(\log{\log{k}})$-approximation to the social welfare for any downward-closed feasibility constraint $\feasible$, where $k$ is the number of maximal feasible sets in $\feasible$.
	\end{theorem} 
	
	\begin{proof}

We fix the vector of bidder values $\vals$.	
	For a set $S \in \maximalSets$ let $\thresh{S}$ be the threshold $p$ such that $\{|\{i \in S : v_i \geq p\}| = 60\log{k}\}$ and $\topS{S}$ denote the subset of bidders with value at least $\thresh{S}$ in $S$ (i.e., the $60\log{k}$ highest value bidders in $S$)\footnote{We assume consistent (lexicographical) tie-breaking in the event that there isn't a price for which exactly $60\log{k}$ bidders have value above the threshold.}. 
	Let $O$ denote the optimal set.  We then have that the WFCA obtains a $O(\log\log{k})$-approximation to the welfare in $\topS{O}$.  Since we run the WFCA with probability $1/2$ in the hedging auction, it just remains to show that the sampling auction obtains a $O(\log\log{k})$-approximation to $v(O) - v(\topS{O})$.  We can upper bound $v(S) - v(\topS{S})$ for any $S$ by using our $S(t, \vals)$ notation from Section \ref{sec:single-price-dc} as $v(S) - v(\topS{S}) = v(S \setminus \topS{S}) \leq \int_{0}^{\thresh{S}}{|S(x,\vals)|dx}$ since every bidder in $S \setminus \topS{S}$ has value no more than $\thresh{S}$ by definition.  To show that the sampling auction obtains a $O(\log\log{k})$-approximation to $\int_{0}^{\thresh{O}}{|O(x,\vals)|dx}$, we use the following lemma.  We defer the proof to Appendix \ref{sec:addlproofs} as it is similar to that of Lemma \ref{lem:chernoff1}. 
	
	\begin{lemma}\label{lem:sampled}
	    When running the sampling auction on any instance we have that for any $S \in \maximalSets$, \[\prob{\exists x \in [0, \thresh{S}] ~:~ |T \cap S(x, \vals)| \notin \left[1/9, 8/9\right]\cdot |S(x, \vals)|} = o(1/k^2).\]
	\end{lemma}
	By taking a union bound over the $k$ maximal sets in $\maximalSets$, we obtain the following corollary.
	\begin{corollary}\label{cor:allsamples}
	    When running the sampling auction on an instance with $k$ maximal sets in $\maximalSets$,
	    \[\prob{\forall S \in \maximalSets ~ \forall x \in [0, \thresh{S}] ~:~ |T \cap S(x, \vals)| \in [1/9, 8/9]\cdot |S(x, \vals)|} = 1 - o(1/k).\]
	\end{corollary}
	
	We can now complete the proof. 
	We first assume that $v(O)/100 \geq \max_S v(\topS{S})$ as otherwise the WFCA would give a $O(\log\log{k})$-approximation.  Let $S^* = \text{argmax}_{S \in \maximalSets}v(S \cap T)$.  We then have that the random sampling auction obtains welfare $v(S^* \cap U)$.  Since $S^*$ was the set of highest sampled welfare, with probability $1-o(1/k)$ the event in Corollary \ref{cor:allsamples} holds and we have that 
	\begin{align*}
	    v(S^* \cap T) \geq v(O \cap T)
	    \geq \int_0^{\thresh{O}}{|T \cap O(x,\vals)|dx}
	    \geq \int_0^{\thresh{O}}{\frac{1}{9}|O(x,\vals)|dx}
	    \geq \frac{1}{9} \cdot (v(O) - v(\topS{O})),
	\end{align*}
	where the third inequality is due to the lower bound on $|T \cap O(x,\vals)|$ from  Corollary \ref{cor:allsamples}. But then, since $v(\topS{O}) \leq v(O)/100$ we have that $v(S^* \cap T) \geq \frac{11v(O)}{100}$. We can bound $v(S^* \cap U)$ similarly by
	\begin{align*}
	    v(S^* \cap U)  &\geq \int_0^{\thresh{S^*}}{|U \cap S^*(x, \vals)|dx}
	    \geq \int_{0}^{\thresh{S^*}}{\frac{1}{9}|S^*(x, \vals)|dx}
	    \geq \frac{1}{9} \cdot (v(S^*) - v(\topS{S^*}))\\
	    &\geq \frac{1}{9} \cdot \left(v(S^*) - \frac{v(O)}{100}\right)
	    \geq \frac{1}{9} \cdot \left(v(S^* \cap T) - \frac{v(O)}{100}\right)
	    \geq \frac{1}{9} \cdot \left(\frac{11v(O)}{100} - \frac{v(O)}{100}\right)
	    \geq \frac{v(O)}{900},
	\end{align*}
	where the second inequality is due to the fact that $|U \cap S^*(x,\vals)| + |T \cap S^*(x, \vals)| = |S^*(x,\vals)|$ and the upper bound on $|T \cap S^*(x,\vals)|$ from  Corollary \ref{cor:allsamples}. Thus, when the WFCA does not give a $O(\log\log{k})$-approximation, the sampling auction gives a $O(1)$-approximation, as desired.
	\end{proof}

\section{Lower Bound}\label{sec:lowerbound}

To complement our positive results, in this section we show that no clock auction (single-price or not) can achieve better than a 29/27-approximation to the expected optimal welfare even under full information and independent values.
This is in contrast to the lower bound in \cite{CGS2022}, which applies only with respect to correlated priors.
Moreover, our bound applies in a very simple feasibility instance that is a special case of both $\dsets$ and $\ksystem$ (which is the simplest setting beyond matroid feasibility, for which an optimal solution can be obtained using a deterministic prior-free clock auction~\cite{MS2019}.)
This result demonstrates that in the process of discovering the high-valued bidders, it may be inevitable that the auction loses some low-valued bidders along the way.

\begin{theorem}\label{th:lowerbound}
For any constant $\epsilon>0$, no clock auction can obtain $29/27 - \epsilon$ approximation to the expected optimal welfare, even in simple $\dsets$ or $\ksystem$ feasibility constraints and even under full access to the independent value distributions.  
\end{theorem}

\begin{proof}
Consider an instance with two disjoint maximal feasible sets $S$ and $T$. $S$ consists of a single bidder with deterministic value $1$. $T$ consists of two bidders, each has value $2/5$ with probability $p = 2/3$ and value $1$ with probability $(1 - p) = 1/3$, independently.  
The expected optimal welfare is $1 \cdot p^2 + 7/5 \cdot 2p(1-p) + 2 \cdot (1-p)^2 = 4/9 + 28/45 + 2/9 = 58/45$.  
We show that no clock auction can achieve a higher expected welfare than $54/45$, concluding the assertion of the theorem.

A clock auction that always serves the bidder in $S$ can only obtain value $1$.  
A clock auction that always serves the bidders in $T$ without increasing either clock above $2/5$ obtains value $4/5 \cdot p^2 + 7/5 \cdot 2p(1-p) + 2 \cdot (1-p)^2 = 16/45 + 28/45 + 2/9 = 54/45$.  
It remains to consider the case where the clock auction increases the clock of at least one bidder in $T$ above $2/5$. 
Fix one of the bidders in $T$ and suppose its clock increases above $2/5$. 
If its value is $2/5$ (this happens with probability $p$), then the auction loses that bidder and can obtain value of at most 1. 
If its value is 1 (this happens with probability $1-p$), the auction obtains value 1 from this bidder, and an expected value $p \cdot 2/5 + (1-p) \cdot 1$ from the other bidder. 
Together, the auction obtains an expected welfare of at most $p \cdot 1 + (1-p)\left(1 + p \cdot 2/5 + (1-p) \cdot 1 \right) = 54/45$, as desired.
\end{proof}

\section{A Hierarchy of Clock Auctions and Additional Results}\label{sec:hierarchy}
Our positive results from the previous sections using simple clock auctions provide hope that clock auctions with prior information can achieve even stronger performance guarantees.
To better understand their power and limitations, we propose the study of these auctions through a hierarchy of gradually increasing power. 

In the figure below, we provide a sketch of this hierarchy among classes of clock auctions, with the arrows pointing from less general to more general classes. The class of \emph{adaptive single-price clock auctions} includes all mechanisms that can combine the strengths of posted-prices mechanisms and single-price clock auctions: an adaptive single-price clock auction can offer a price to each buyer in any desired order, and the price offered to each buyer can adapt to the responses provided by other buyers, earlier in the ordering. Furthermore, the decision regarding which subset of buyers should be served can be deferred until all buyers have responded to their offers.

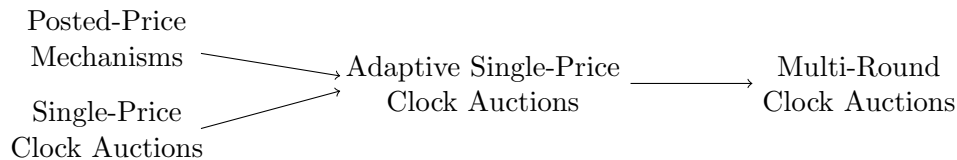
\begin{center}
	\begin{tikzpicture}
	\node[align=center] (A) at (0,0) {Single-Price \\ Clock Auctions};
	\node[align=center] (B) at (0,1.2) {Posted-Price \\ Mechanisms};
	\node[align=center] (C) at (5,0.6) {Adaptive Single-Price \\ Clock Auctions};
	\node[align=center] (D) at (10,0.6) {Multi-Round \\ Clock Auctions};
	\draw [->] (1.25,0) -- (3.1,0.5);
	\draw [->] (1.25,1) -- (3.1,0.7);
	\draw [->] (C) edge (D);
	\end{tikzpicture}
\end{center}

In the remainder of this section, we first present some initial results and observations for \caucts\ (\S \ref{subsec:SQCA}). We then propose the refinement of multi-round clock auctions based on the number of distinct prices offered to each bidder, or the extent to which the choice of the prices can adapt to previous bidder responses (\S \ref{subsec:genBCA}).

\subsection{Single-Price Clock Auctions}\label{subsec:SQCA}

\vspace{0.05in}
\noindent {\bf Binary valuations.}
A class of instances that emphasizes the benefits of single-price clock auctions over posted-price mechanisms is that of binary valuations, i.e., where the value of each bidder can be either 0 or 1. Although this class is seemingly simple, it actually poses some of the most significant obstacles for posted-price mechanisms. Many of the lower bounds in the prophet inequality literature come from this class, and even for the special case of \dsets\ with $n$ binary value bidders, one cannot obtain a better than $O(\log n/\log\log n)$ approximation \cite{KW2012, R2016}. This is in stark contrast to single-price clock auctions, which can trivially achieve the optimal social welfare:
\begin{observation}
For the special case of binary valuations (i.e., $v_i\in\{0,1\}$ for all $i\in N$), every downward-closed set system admits a single-price clock auction that achieves $\opt$. 
\end{observation}
Indeed, the auction can just set a price of, say, $1/2$ for all bidders\footnote{Clearly, any price $p\in (0, 1)$ would achieve the same result.}, at which point the bidders whose value is 0 will drop out, revealing the bidders with value 1. Then, the auction can choose to serve the largest feasible subset of the remaining bidders, leading to the optimal social welfare. 

\vspace{0.1in}
\noindent {\bf Distributions with bounded support size.}
A more general class of instances where a single-price clock auction performs well is instances with bounded support size. Specifically, if the support size $\supp(\disti)$ is at most $\ell$ for every $i\in N$, then we can achieve an approximation of $\ell$ (proof deferred to  Appendix~\ref{sec:proof_types}). \begin{restatable}{claim}{constantSupport}
Suppose $|\supp(\disti)|\le\ell$ for all $i\in N$, then there is a single-price clock auction using the full distributional prior $\dist$ that obtains $\ell$ approximation to the expected optimal social welfare.
\end{restatable}

\vspace{0.1in}
\noindent {\bf Connections to large deviation theory.}
In cases beyond constant support size and binary valuations, the question of whether single-price clock auctions can achieve a constant approximation is quite non-trivial. In fact, the problem remains non-trivial even if we restrict attention to the very special case of $\dsets$ with equal size maximal sets (i.e., $k$ disjoint maximal sets of $m=n/k$ bidders each) and i.i.d.\ bidders (i.e., $\disti=\distone$ for all $i\in N$). Interestingly, this question appears to be  related to the {\em large deviation theory} \cite{den2008large}.

For instances where all $k$ maximal disjoint sets have the same size $m$, and all $n=km$ bidders have their $\vali$ values drawn i.i.d., 
the expected optimal welfare is equal to the maximum taken over $k$ independent samples of $\xi(m)=\sum_{i=1}^{m}\vali$. 
We are interested in the case where this maximum is significantly larger than the expectation of $\Ex{\xi(m)}$ which only happens for large $k =\Theta(2^m)$. In this case, the expectation of this maximum is roughly equal to the value of the top $k$-th quantile of $\xi(m)$. 
Note that in the limit when $k =\Theta(2^m)$ and $m\to\infty$, the tail probability for the distribution of $\xi(m)$ is captured by Cramer's theorem (a fundamental result in large deviation theory). According to this theorem, there is a certain rate function $I$ such that $\xi(m)$ satisfies a \emph{large deviation principle with rate function $I$}, meaning that 
$\prob{\frac{\xi(m)}{m}\ge x}=e^{-I(x)\cdot m\cdot (1+o(1))}$.

On the other hand, in a single-price clock auction with price vector $\prices$, the agents with $\vali<\pricei$ drop out, so the event that $\vali\sim\disti$ lies in the ``tail'' of $\disti$ ($\vali\ge\pricei$) is instead captured by the Bernoulli random variable $\valhi=\ind{\vali\ge\pricei}\cdot\Ex{\vali~\vert~ \vali\ge\pricei}$. Given the set of bidders that accepted the prices offered to them, the obvious choice is then to serve the feasible subset with the largest sum of prices, leading to an expected social welfare of $\max_{S\in\maxsets}\sum_{i\in S}\valhi$. 

We can think of the variable $\valhi$ as a {\em binary approximation} of the random variable $\vali$ up to an additive error\footnote{As a single-price clock auction with $0$ prices achieves in expectation the welfare of $m\cdot\Ex{\vali}$, we can safely ignore the additive terms of order $O(\Ex{\vali})$ for each bidder $i$.} of $O(\Ex{\vali})$. Formally, a {\em binary approximation} of a random variable $\xi_i:\Omega_i\to\reals$ on a probability space $\Omega_i$ is the following (generalized) Bernoulli random variable $\oxi_i$ defined for an event $\ev_i\subset\Omega_i$.
\[
\oxi_i:\quad \Ex{\xi_i ~\vert~ \ev_i} ~\text{w.p.}~\prob{\ev_i}\quad,\quad
\Ex{\xi_i~ \vert ~\overline{\ev}_i} ~\text{w.p.}~\prob{\overline{\ev}_i}
\]
Then, to achieve $O(1)$ approximation single-price clock auctions, we need to cover the  tail probabilities of $\frac{1}{m}\sum_{i=1}^m\vali$ in the large deviation regimes by the constantly related tail probabilities of $\frac{1}{m}\sum_{i=1}^m\valhi$. In the language of large deviation theory, we need to cover the tail probability of the sum of i.i.d. random variables by the tail probability of the sum of binary approximations.
This is captured in the following claim.
\begin{claim}
\label{cl:large_deviations}
Assume that there is a single-price clock auction for the $\dsets$ with $k$ sets of size $m$ and $n=k\cdot m$ i.i.d. bidders that is $O(1)$ approximation $\forall k,m\in\natural$. 
Then for any collection $(\xi_i)_{i=1}^{m}$ of $m$ non-negative i.i.d. random variables and for every $x > O(\Ex{\xi_i})$ there must exist binary approximations $(\oxi_i)_{i=1}^{m}$ of $(\xi_i)_{i=1}^{m}$ such that 
\[
\prob{\frac{\sum_{i=1}^m\oxi_i}{m}\ge x\cdot\Omega(1)}\ge \prob{\frac{\sum_{i=1}^m\xi_i}{m}\ge x}.
\] 
\end{claim}
We note that the binary approximations $\oxi_i$ in Claim~\ref{cl:large_deviations} may not be identical. If, on the other hand, we could achieve the constant approximation using a uniform single-price clock auction (which is reasonable to expect in the i.i.d.\ setting), then this would imply a stronger Claim~\ref{cl:large_deviations} 
with identical binary approximations $(\oxi_i)_{i=1}^{m}$.

\subsection{Multi-Round Clock Auctions} \label{subsec:genBCA}
As we demonstrate in Sections \ref{sec:dc} and \ref{sec:rand-clock}, moving beyond single-price clock auctions opens up an even greater set of design options.  Gradually increasing prices over time allows a clock auction to more carefully discover the preferences of the bidders.  In particular, in Section \ref{sec:dc}, we leverage multiple rounds of uniform prices to reduce the prior information required, and, in Section \ref{sec:rand-clock}, we leverage randomization and multiple rounds of non-uniform prices to eschew the need for prior information altogether.  Examining the trade-off between auction complexity versus prior information access is a fundamental question that follows naturally from our work.

An interesting example where multiple rounds of prices are clearly more powerful than a single price with distributional information is the case of selling a single item: the well-known second-price auction is implementable as an ascending price clock auction (with uniform prices for all bidders as in our auction from Section \ref{sec:dc}).  In this case, even without prior information, a multi-round clock auction obtains the optimal welfare.  On the other hand, a single-price clock auction with access to full prior information cannot obtain the optimal welfare in the single item setting even if bidders' values are drawn from i.i.d. distributions.  The relative power of offering multiple rounds of prices versus distributional information also then seems to depend on the feasibility constraint in question, suggesting a rich and fundamental line of further research.

The multi-round clock auctions we describe in this work use prior information in a very minimal way (or use no information at all). General multi-round clock auctions with full distributional information, however, can follow highly complicated price trajectories which leverage this information without sacrificing any of their appealing properties. One could then use this price discovery process to update the posterior beliefs regarding bidder values over time to possibly obtain even better performance as a function of the underlying ``complexity'' of the clock auction itself.

We believe that an interesting way to gradually evaluate the power of more complicated clock auctions is through an even more refined hierarchy depending on (i) the number of prices offered to each bidder, and (ii) whether or not they can adapt these queries to the responses of the bidders.  The class of clock auctions maintains its appealing properties no matter how many prices they offer to the bidders, yet, as we reflect in this work, increasing the number of prices comes at a cost.  Minimizing the number of prices offered to each bidder reduces the burden on them and speeds up the auction which may be desirable from an implementation perspective, further motivating a deeper analysis of the hierarchy we propose.

\section{Conclusions and Future Work}
In this paper, we achieve significantly improved social welfare approximation guarantees using Bayesian and randomized clock auctions; our results suggest many directions for future work. Most notably, understanding the best approximation guarantees that can be obtained by deterministic Bayesian clock auctions or randomized prior-free clock auctions under various feasibility constraints is an open line of research.  Even for the conceptually simple $\dsets$ setting, there is an asymptotic gap between the best known upper and lower bounds.  A particularly exciting question is whether or not clock auctions can achieve constant approximation guarantees for all downward-closed settings, via access to priors or randomization.  

Another direction for future research is toward a refined understanding of the impact that limited prior information has on the performance of clock auctions. For example, how would the results change if the auctioneer has only a limited number of samples from the value distribution of each bidder?
This direction has been studied in the prophet inequality literature (see, e.g.,~\cite{AKW19,correa2019prophet,RWW2020}).

Finally, it would be interesting to gain a better understanding of the hierarchy of clock auctions discussed in Section~\ref{sec:hierarchy}. 
For example, is there an asymptotic separation between the best approximation that can be guaranteed by general clock auctions with full distributional priors and single-price clock auctions (with the same distributional information) in various settings? Moreover, one can define a hierarchy of clock auctions based on the number of different prices they are allowed to offer to each bidder. Results in this paper are obtained at the two extremes: single-price clock auctions or clock auctions that use an arbitrary number of rounds. What is the power of clock auctions that are allowed to offer up to $q>1$ prices to each bidder?

\newpage
\appendix
\section{Omitted Proofs}\label{sec:addlproofs}
\subsection{Proof of Lemma \ref{lem:singlefixedprice} and Corollary \ref{lem:fixedprice}}
We note that the proofs of Lemma \ref{lem:singlefixedprice} and Corollary \ref{lem:fixedprice} are the same.  We thus provide one unified proof for both below.
\begin{proof}
Let $m=10\log k+1$ and $\Delta=\Ex{\max_{S\in\maxsets}\InParentheses{\sum_{i\in S}\valhi \cdot\ind{|S(t_S,\vals)|\le m}}}$.
Consider an event that there exists a bidder with value greater than or equal to $2\Delta$.  
By posting the price of $2\Delta$ for every bidder, 
we identify all such bidders and would serve at least one if they exist.  
Following the same reasoning as in \cite{R2016}, let $q_i$ be the probability that bidder $i$ has value $\valhi \ge 2\Delta$.  Then 
\[
\Delta \geq 2 \cdot \Delta \cdot \prob[\vals]{\exists i: \valhi \ge 2\Delta} = 2\cdot \Delta \cdot \left(1 - \prod_{i}(1 - q_i)\right).
\] 
Rearranging gives $\frac{1}{2} \leq \prod_{i}(1 - q_i) \leq e^{-\sum_{i}{q_i}}$ but then $\sum_{i}{q_i} \leq \ln{2}.$  The probability that we would want to accept a bidder and cannot (because there is some other bidder in another feasible set we are taking) is less than $\ln{2}$ so we must obtain at least a $(1 - \ln{2})$-fraction of the contribution from bidders with value above $2 \cdot \Delta$ by serving at most one such bidder.

Consider how much the small bidders with value less than $2\cdot \Delta$ contribute to the benchmark. The approximation guarantee is achieved by the revenue (not the welfare) of a \cauct.  Observe that the contribution to the benchmark of bidders with $\valhi\le\frac{\Delta}{2m}$ is at most $\frac{\Delta}{2}$ so by ignoring all such bidders we lose at most a constant factor against the benchmark.  Consider partitioning all remaining bidders into $\log{m} + 2$  buckets $\left[\frac{\Delta}{2m}, \frac{\Delta}{m}\right], \left[\frac{\Delta}{m}, \frac{2\Delta}{m}\right], \cdots, \left[\Delta, 2\Delta\right]$ 
based on their value.  We can decompose the remaining half of the benchmark (among bidders with value less than $2 \cdot \Delta$) by upper bounding it by the sum of the best set for each of these buckets.  By the pigeonhole principle, the largest expected contribution of these buckets is a $O\left(\log{m}\right)$ fraction of the contribution to the benchmark among bidders with value less than $2 \cdot \Delta$.  Consider the best of these buckets $B$.  The revenue obtained by posting the lower bound of the interval for $B$ is a $2$-approximation to the welfare contained in $B$ for each set (by construction) and, in particular, the one with maximum realized value.  Thus, posting the price equal to the lower bound of $B$ for each bidder and then selecting the set with the greatest number of accepting bidders obtains revenue which is a $O\left(\log{m}\right)$ fraction of the expected contribution from values less than or equal to $2 \Delta$.  Finally, choosing the better of the core (prices below $2\Delta$) and the tail ($p=2\Delta$) then gives a $O\left(\log m\right)$-approximation to $\Delta$.
\end{proof}

\subsection{Proof of Lemma \ref{lem:sampled}}
\begin{proof}
	Consider an arbitrary set $S$ and the threshold $\thresh{S}$ for which $|S({\thresh{S}}, \vals)| = 60\log{k}$.  Since in the sampling auction each bidder is sampled independently with probability $1/2$ we have that $\Ex{T \cap S(\thresh{S}, \vals)} = 30\log{k} = \Ex{U \cap S(\thresh{S}, \vals)}$.  Now consider the threshold $t'_S$ such that $|S(t'_S, \vals)| = 80\log{k} = 4/3\cdot|S(\thresh{S}, \vals)|$.  Similarly as above, we have that $\Ex{T \cap S(t'_S, \vals)} = \Ex{U \cap S(t'_S,\vals)} = 40\log{k}$.  By a Chernoff bound, it follows that
	
	\begin{equation}
		\prob{|T \cap S(t'_S, \vals)| > \frac{4}{3} \Ex{|T \cap S(t'_S, \vals)|}} < \left(\frac{e^{1/3}}{(4/3)^{4/3}}\right)^{40\log{k}} < 0.951^{40\log{k}} = o(1/k^{2}).
	\end{equation}
 
	Now, for all $p \in [t'_S, \thresh{S}]$ we know that $|T \cap S(p, \vals)| \leq |T \cap S(t'_S, \vals)|$ and $\Ex{|T \cap S(p, \vals)|} \geq \Ex{|T \cap S(\thresh{S}, \vals)|} = \frac{3}{4} \cdot \Ex{|T \cap S(t'_S, \vals)|}$.  Therefore, if $|T \cap S(t'_S, \vals)| < \frac{4}{3}\Ex{|T \cap S(t'_S, \vals)|}$, then $|T \cap S(p, \vals)| < \frac{16}{9}\Ex{|T \cap S(p, \vals)|}$ for all $p \in [t'_S, \thresh{S}]$.  Hence, we have
	\begin{equation}
		\prob{\exists p \in [t'_S, \thresh{S})~:~|T \cap S(p, \vals)| > \frac{16}{9} \Ex{|T \cap S(p, \vals)|}}  = o(1/k^{2}).
	\end{equation}

	Consider intervals $[r_0, r_1), [r_1, r_2), \dots, [r_{\ell - 1}, r_\ell)$ where $\Ex{|T \cap S(r_j, \vals)|} = \frac{3}{4} \cdot \Ex{|T \cap S(r_{j-1}, \vals)|}$ for all $j = 1, \dots, \ell$ and such that $r_0 = 0$ and $r_\ell = \thresh{S}$.  We can apply the same Chernoff bound argument that we did for the interval $[t'_S, \thresh{S})$ to any such interval and get a bound of $o\left(1/k^{2\cdot(4/3)^{\ell - j + 1}}\right)$ for the interval $[r_{j-1}, r_j)$.  But then, we may take a union bound over all of these intervals to obtain that
	\begin{equation*}
		\prob{\exists p \in \cup_{j = 1}^{j = \ell} [r_{j-1}, r_j)~:~|T \cap S(p, \vals)|  > \frac{16}{9}\cdot \Ex{|T \cap S(p, \vals)|}} \leq \sum_{j = 1}^{j = \ell}{o\left(\frac{1}{k^{2\cdot(4/3)^{\ell - j + 1}}}\right)} = o\left(\frac{1}{k^2}\right).
	\end{equation*}
	Consequently, $\prob{\exists p \in [0, \thresh{S})~:~|T \cap S(p, \vals)|  > \frac{16}{9}\cdot \Ex{|T \cap S(p, \vals)|}} = o(1/k^2)$. Since we have $\Ex{|T \cap S(p, \vals)|} = \frac{1}{2} \cdot |S(p, \vals)|$ we obtain $\prob{\exists p \in [0, \thresh{S})~:~|T \cap S(p, \vals)|  > \frac{8}{9}\cdot |S(p, \vals)|} = o(1/k^2)$.  We can repeat the argument replacing the sampled set of bidders $T$ with the unsampled set $U = N \setminus T$ and obtain that $\prob{\exists p \in [0, \thresh{S})~:~|U \cap S(p, \vals)|  > \frac{16}{9}\cdot \Ex{|U \cap S(p, \vals)|}} = o(1/k^2)$ and thus $\prob{\exists p \in [0, \thresh{S})~:~|U \cap S(p, \vals)|  > \frac{8}{9}\cdot |S(p, \vals)|} = o(1/k^2)$.  Finally, since $U \cap S(p,\vals) = S(p, \vals) \setminus (T \cap S(p,\vals))$, we obtain $\prob{\exists p \in [0, \thresh{S})~:~|T \cap S(p, \vals)|  < \frac{1}{9}\cdot |S(p, \vals)|} = o(1/k^2)$, completing the proof.
	\end{proof}

\subsection{Proof of Theorem \ref{th:k_sets}}
\begin{proof}
Our proof proceeds similarly as the one in Section \ref{sec:single-price-dc}, again employing our benchmark decomposition, inequality \eqref{eq:low_high_decomposition}.  Note that $\auco$ is exactly the single-price clock auction with all prices equal to zero, which, from Section \ref{sec:single-price-dc}, we know covers the $\lowcore$, $\lowtail$, and $\hightail$ portions of our benchmark.  In other words, from our analysis in the previous section, we have the following two claims:
\begin{claim} $\hightail\le \Ex[\vals]{\auco(\vals)}=\max_{S\in\maxsets}\sum_{i\in S}\Ex{\vali}$
\label{cl:high-tail}
\end{claim}
and
\begin{claim}
$\mathsf{LOW}\le 12\cdot\Ex[\vals]{\auco(\vals)}=12\cdot \max_{S\in\maxsets}\sum_{i\in S}\Ex{\vali}$.
\label{cl:low}
\end{claim} 

Inequalities \eqref{eq:low_high_decomposition} and \eqref{eq:high_core_tail_decomposition} give us an upper bound $\opt\le \mathsf{LOW}+ \hightail+\highcore$, where the welfare of $\auco$ covers the terms $\mathsf{LOW}$ and $\hightail$ by claims~\ref{cl:high-tail} and~\ref{cl:low}: $13\cdot\Ex[\vals]{\auco(\vals)}\ge\mathsf{LOW}+ \hightail$. Thus

\[
\opt - 13\cdot \Ex[\vals]{\auco(\vals)} \le \highcore \le \opt.
\]

We can now assume that $\highcore$ is a $2$ approximation to $\opt$, as otherwise $\auco$ would be a $26 < \log{\log{k}}$ approximation to $\opt$. Thus it suffices to show that $\uprice$ gives a $O(\log\log k)$ approximation to $\highcore$. 
From the proof of Lemma \ref{lem:singlefixedprice} we obtain Corollary \ref{lem:fixedprice} which shows that there is a single price $p$ such that either our {\em revenue} will be a $O(\log\log k)$ approximation to $\highcore$ when we uniformly post $p$ for each bidder $i\in N$ and serve a maximal feasible set of accepting bidders, or the price $p$ is so high that by just serving 
a single bidder we can get the expected {\em welfare} 
to be a $O(\log\log k)$ approximation to $\highcore$.

Note that 
$\price\cdot\Ex{\max_{S\in\maxsets}\sum_{i\in S}\ind{\vali\ge\price}}$
is precisely the expected revenue
we get in a clock auction with a uniform price $p$.
Corollary~\ref{lem:fixedprice} shows that the best choice of a single uniform price 
will give an $O(\log m)$-approximation to $\Delta$, which translates to $O(\log\log k)$-approximation to $\highcore$ with $m = 10 \log{k} + 1$. The proof of this corollary directly follows the proof of Lemma \ref{lem:singlefixedprice}.

Note that Corollary~\ref{lem:fixedprice} proves the existence of a uniform price $p$ that would give us the desired approximation of $O(\log\log k)$. Since we cannot calculate this price in advance, $\uprice$ gradually sweeps through a full range of prices, aiming to find one that yields the desired revenue. However, this could fail if $\uprice$ reaches a lower price that also satisfies the revenue target and stops too soon, selecting a set with low social welfare.
We get around this problem by more careful analysis of the expected welfare of $\uprice$. 
We choose revenue goal $g=\frac{\opt}{4\alpha}$ in $\uprice$ based on the approximation guarantee $\alpha$ in Corollary~\ref{lem:fixedprice}. 
Recall that  
\begin{equation} 
\label{eq:def_g}
g=\frac{\opt}{4\alpha}\le \frac{\highcore}{2\alpha}.
\end{equation}

We consider three cases for $\uprice$. Let $T(\vals)$ be a feasible set of maximum size among the bidders who accept the price $p$ from Corollary~\ref{lem:fixedprice}.
\paragraph{Case 1: $\Ex[\vals]{|T(\vals)|}\ge 8$, i.e., the expected size of the set $T$ is large.} 
Note that 
\[|T(\vals)|=\max_{S\in\maxsets}\sum_{i\in S}X_i, ~\text{where}~ X_i\eqdef\ind{\vali\ge p}~\text{are independent Bernoulli r.v. for } i\in N.
\]
Thus, one can think of $|T(\vals)|$ as an XOS (maximum of additive) function $f(X_1,\ldots,X_n)$ with marginal contribution of each coordinate $i$ in $[0,1]$. According to~\cite{Vondrak10}, such functions enjoy strong concentration with Chernoff-type bounds. In particular, Corollary~$3.2$ in~\cite{Vondrak10} gives us 
\[
\prob{|T|\le\Ex{|T|}/2}\le e^{-\Ex{|T|}/8}\le e^{-1}.
\] 
Let us assume that $|T(\vals)|\ge\Ex{|T|}/2$.
According to Corollary~\ref{lem:fixedprice} and \eqref{eq:def_g} we have
 $2g\le \frac{\highcore}{\alpha}\le\price\cdot\Ex{|T|}$, i.e., $g\le\price\cdot|T(\vals)|$. In this case, $\uprice$ obtains welfare of at least $g$. Indeed, the auction must stop at price $p$ or earlier, as the revenue goal of $g$ would have been attained at price $p$. Note that the auction can stop only when the revenue goal is reached, or all active bidders can be served. In the former case, the revenue of $\uprice$ is at least $g$. In the latter case, all the bidders in $T(\vals)$ will be served before reaching the price $p$, i.e., everyone in $T(\vals)$ with  the welfare of at least $g$. Therefore, the expected social welfare of $\uprice$ is at least
 \[
\prob{|T|\ge\Ex{|T|}/2}\cdot g \ge (1-e^{-1})g =\frac{\opt}{O(\log\log k)}. 
 \]

\paragraph{Case 2: $\prob[\vals]{\uprice \text{ serves set } F: \sum_{i\in F}\vali\ge g}\ge \frac{1}{2}$.} Then the expected welfare of $\uprice$ is at least  $g/2 = \frac{\opt}{O(\log\log k)}$.
 
\paragraph{Case 3: $\Ex[\vals]{|T(\vals)|}< 8$ and $\prob[\vals]{\uprice \text{ serves set } F: \sum_{i\in F}\vali\ge g}< \frac{1}{2}$.} 
Thus, we stop at set $F$ with $\sum_{i\in F}\vali< g$ with probability more than $1/2$. When $\sum_{i\in F}\vali< g$, $\uprice$ must be serving the whole set $A$ of active bidders. In particular, it means that a bidder $j\in\argmax_{i\in N}\vali$ must be served in $A$, as we can assume without loss of generality that each bidder in $N$ can be served as a singleton. We show in the following Claim~\ref{cl:u-price_expected_max} that the expected welfare of $\uprice$ must be at least half of the {\em expected maximum}.\footnote{Despite the fact that we serve the maximum bidder with probability more than $1/2$, we still need to show that $\uprice$ gets sufficiently high expected welfare.} 

\begin{claim} 
\label{cl:u-price_expected_max}
Assume $\prob[\vals]{\uprice \text{ serves set } F: \sum_{i\in F}\vali\ge g}< \frac{1}{2}$. Then the expected social welfare of $\uprice$ is at least $\frac{1}{2}\Ex[\vals]{\max_i\vali}.$
\end{claim}
\begin{proof}
Sample $\vals\sim\vect{\dist}$ and fix the bidder $i=\argmax_j \vali[j]$ with the maximum value $\vali=t$.
Assuming that $i$ is the bidder with maximum value $t$, it suffices to show that $\uprice$ serves $i$ with probability at least  $1/2$ for any $i\in N$ and $t\ge 0$. To this end, consider independently sampling valuation profile $\pvalsmi\sim\pdistsmi: \forall j\ne i~ \pvali[j]\sim \disti[j] \mid \pvali[j] < t$. Note that (i) $\pvalsmi$ has exactly the same distribution as $\valsmi$ conditioned on $\vali=t$ being the maximum in $\vals$ and (ii) $\pdistsmi$ is stochastically dominated by $\walsmi\sim\distsmi=\times_{j\ne i} \disti[j]$, i.e., we can do probability coupling  between $\pvalsmi\sim\pdistsmi$ and $\walsmi\sim\distsmi$, so that each $\pvalsmi$ is coordinate-wise smaller than $\walsmi$ ($\forall j\ne i,~\pvali[j]\le \wali[j]$).

By the Claim's assumption the social welfare of $\uprice$ on a random valuation profile $(\walsmi,\wali)\sim\vect{\dist}$ is less than $g$ with probability more than $1/2$. Now, if this welfare is less than $g$, then the revenue of offering a uniform price $q$ must  be smaller than $g$ for any price $q\ge 0$, i.e.,
\begin{equation}
\label{eq:must_serve_i}
\forall q\ge 0,~~  q\cdot|\{j\in N: \wali[j]\ge q\}| <g \quad\implies\quad \forall q\ge 0,~~  q\cdot|\{j\ne i: \pvali[j]\ge q\}| <g.
\end{equation}
Therefore, $\uprice$ must serve bidder $i$ on the valuation profile $(\pvalsmi, t)$, as (a) $i$ has the maximum bid and would always stay active; (b) we cannot reach the revenue goal of $g$ with a uniform price $q$ without using bidder $i$ according to \eqref{eq:must_serve_i}. This concludes the proof, as the probability of having 
a $(\walsmi,\wali)\sim\vect{\dist}$ such that the welfare of  $\uprice(\walsmi,\wali)<g$ is greater than half and  we will serve $i$ on each corresponding profile $(\pvalsmi,t)$.
\end{proof}
Now we can conclude the analysis of case 3. Claim~\ref{cl:u-price_expected_max} gives a lower bound of $\frac{1}{2}\Ex[\vals]{\max_i \vali}$ on the  expected welfare of $\uprice$. When $\Ex{|T|}<8$, the expected maximum $\Ex[\vals]{\max_i \vali}$ is a constant approximation to the revenue $p\cdot\Ex{|T|}$ of $\uprice$. 
On the other hand,  $\frac{1}{2}\Ex[\vals]{\max_i \vali}$ is also a constant approximation to the other term  $\Ex{\max_{i\in N}\Ex{\vali ~\big|~ \vali\ge\price}\cdot\ind{\vali\ge\price}}$ in Lemma~\ref{lem:fixedprice}.
\end{proof}

\subsection{Single-price clock auctions for the distributions with constant support sizes}
\label{sec:proof_types}
Let $\feasible$ be any downward-closed family of feasible sets over the set of $N$ bidders. We assume that all possible values $\vali$ of each bidder $i\in N$ can be described by a list of size at most $\ell$, i.e., the support $\supp(\disti)$ of each distribution $\disti$ has size  at most $\ell$.
\constantSupport*
\begin{proof}
Let us write all possible types in each distribution $\disti$ in the increasing order $(\theta_{i,1},\ldots, \theta_{i,\ell})$. Consider the following $\ell$ single-price clock auctions: for a given $j\in \{1,2,\ldots,\ell\}$ 
\begin{enumerate}
		\item set price $\pricei=\theta_{i,j}$ for each $i\in N$.
		\item among the bidders who accepted their price select $S\subset\{i:\vali\ge\theta_{i,j}\}$ with maximal revenue 
		\[\max_{S\in\feasible}\sum_{i\in S}\theta_{i,j}.\]
\end{enumerate}
Our single-price clock auction is the one among these $\ell$ auctions with the highest expected revenue\footnote{Note that we can write the revenue of the auction in this way because $\feasible$ is downward-closed set system.}, which can be written as 
\[
\max_{j\in[\ell]}\Ex[\vals]{\max_{S\in\feasible}\sum_{i\in S}\theta_{i,j}\cdot\ind{\vali\ge\theta_{i,j}}}.
\] 
We have the following upper bound on the benchmark.
\begin{multline}
\label{eq:types}
\opt=\Ex[\vals]{\max_{S\in\feasible}\sum_{i\in S}\vali}
=\Ex[\vals]{\max_{S\in\feasible}\sum_{i\in S}\vali\sum_{j=1}^{\ell}\ind{\vali=\theta_{i,j}}}
\le
\Ex[\vals]{\sum_{j=1}^{\ell}\max_{S\in\feasible}\sum_{i\in S}\vali\cdot\ind{\vali=\theta_{i,j}}}\\
=\sum_{j=1}^{\ell}\Ex[\vals]{\max_{S\in\feasible}\sum_{i\in S}\theta_{i,j}\cdot\ind{\vali=\theta_{i,j}}}
\le\ell\cdot \max_{j\in[\ell]}\Ex[\vals]{\max_{S\in\feasible}\sum_{i\in S}\theta_{i,j}\cdot\ind{\vali=\theta_{i,j}}},
\end{multline}
where to get the first inequality we used that for each valuation profile $\vals$ the maximal sum $\sum_{i\in S}\vali$ for $S\in\feasible$can be covered by $\ell$ sums $\sum_{i\in S} \vali\cdot\ind{\vali=\theta_{i,j}}$. To conclude the proof we notice that for each valuation profile $\vals$ and feasible $S\in\feasible$ \[\sum_{i\in S}\theta_{i,j}\cdot\ind{\vali=\theta_{i,j}}\le \sum_{i\in S}\theta_{i,j}\cdot\ind{\vali\ge\theta_{i,j}},\]
i.e., the last term in \eqref{eq:types} is not more than $\ell$ times the revenue of our single-price clock auction, which is not more than the welfare. 
\end{proof}
\newpage
\bibliographystyle{plainnat}
\bibliography{bibliography}

\begin{thebibliography}{55}
\providecommand{\natexlab}[1]{#1}
\providecommand{\url}[1]{\texttt{#1}}
\expandafter\ifx\csname urlstyle\endcsname\relax
  \providecommand{\doi}[1]{doi: #1}\else
  \providecommand{\doi}{doi: \begingroup \urlstyle{rm}\Url}\fi

\bibitem[Abolhassani et~al.(2017)Abolhassani, Ehsani, Esfandiari, Hajiaghayi,
  Kleinberg, and Lucier]{AEEHKL2017}
Melika Abolhassani, Soheil Ehsani, Hossein Esfandiari, MohammadTaghi
  Hajiaghayi, Robert Kleinberg, and Brendan Lucier.
\newblock Beating 1-1/e for ordered prophets.
\newblock In \emph{Proceedings of the 49th Annual ACM SIGACT Symposium on
  Theory of Computing}, pages 61--71, 2017.

\bibitem[Akbarpour and Li(2020)]{AL2020}
Mohammad Akbarpour and Shengwu Li.
\newblock Credible auctions: A trilemma.
\newblock \emph{Econometrica}, 88\penalty0 (2):\penalty0 425--467, 2020.

\bibitem[Azar et~al.(2013)Azar, Micali, Daskalakis, and Weinberg]{AMDW2013}
Pablo Azar, Silvio Micali, Constantinos Daskalakis, and S~Matthew Weinberg.
\newblock Optimal and efficient parametric auctions.
\newblock In \emph{Proceedings of the Twenty-Fourth Annual ACM-SIAM Symposium
  on Discrete Algorithms}, pages 596--604. SIAM, 2013.

\bibitem[Azar and Micali(2013)]{AM13}
Pablo~Daniel Azar and Silvio Micali.
\newblock Parametric digital auctions.
\newblock In Robert~D. Kleinberg, editor, \emph{Innovations in Theoretical
  Computer Science, {ITCS} '13, Berkeley, CA, USA, January 9-12, 2013}, pages
  231--232. {ACM}, 2013.

\bibitem[Azar et~al.(2019)Azar, Kleinberg, and Weinberg]{AKW19}
Pablo~Daniel Azar, Robert Kleinberg, and S.~Matthew Weinberg.
\newblock Prior independent mechanisms via prophet inequalities with limited
  information.
\newblock \emph{Games Econ. Behav.}, 118:\penalty0 511--532, 2019.

\bibitem[Babaioff et~al.(2007)Babaioff, Immorlica, and Kleinberg]{BIK2007}
Moshe Babaioff, Nicole Immorlica, and Robert Kleinberg.
\newblock Matroids, secretary problems, and online mechanisms.
\newblock In \emph{Proceedings of the Eighteenth Annual ACM-SIAM Symposium on
  Discrete Algorithms}, pages 434--443, USA, 2007. Society for Industrial and
  Applied Mathematics.
\newblock ISBN 9780898716245.

\bibitem[Babaioff et~al.(2020)Babaioff, Immorlica, Lucier, and
  Weinberg]{BILW2014}
Moshe Babaioff, Nicole Immorlica, Brendan Lucier, and S~Matthew Weinberg.
\newblock A simple and approximately optimal mechanism for an additive buyer.
\newblock \emph{Journal of the ACM (JACM)}, 67\penalty0 (4):\penalty0 1--40,
  2020.

\bibitem[Balkanski et~al.(2022)Balkanski, Garimidi, Gkatzelis, Schoepflin, and
  Tan]{BGGST22}
Eric Balkanski, Pranav Garimidi, Vasilis Gkatzelis, Daniel Schoepflin, and
  Xizhi Tan.
\newblock Deterministic budget-feasible clock auctions.
\newblock In \emph{Proceedings of the Thirty-Third Annual ACM-SIAM Symposium on
  Discrete Algorithms}, SODA 2022, USA, 2022. Society for Industrial and
  Applied Mathematics.

\bibitem[Beyhaghi et~al.(2021)Beyhaghi, Golrezaei, Leme, P{\'a}l, and
  Sivan]{BGLPS2018}
Hedyeh Beyhaghi, Negin Golrezaei, Renato~Paes Leme, Martin P{\'a}l, and
  Balasubramanian Sivan.
\newblock Improved revenue bounds for posted-price and second-price mechanisms.
\newblock \emph{Operations Research}, 69\penalty0 (6):\penalty0 1805--1822,
  2021.

\bibitem[Bichler et~al.(2020)Bichler, Hao, Littmann, and
  Waldherr]{bichler2020strategyproof}
Martin Bichler, Zhen Hao, Richard Littmann, and Stefan Waldherr.
\newblock Strategyproof auction mechanisms for network procurement.
\newblock \emph{OR Spectrum}, 42\penalty0 (4):\penalty0 965--994, 2020.

\bibitem[Brandt and Sandholm(2005)]{BS05}
Felix Brandt and Tuomas Sandholm.
\newblock Unconditional privacy in social choice.
\newblock In Ron van~der Meyden, editor, \emph{Proceedings of the 10th
  Conference on Theoretical Aspects of Rationality and Knowledge (TARK-2005),
  Singapore, June 10-12, 2005}, pages 207--218. National University of
  Singapore, 2005.

\bibitem[Chawla et~al.(2010)Chawla, Hartline, Malec, and Sivan]{CHMS2009}
Shuchi Chawla, Jason~D Hartline, David~L Malec, and Balasubramanian Sivan.
\newblock Multi-parameter mechanism design and sequential posted pricing.
\newblock In \emph{Proceedings of the Forty-Second ACM Symposium on Theory of
  Computing}, pages 311--320, 2010.

\bibitem[Christodoulou et~al.(2022)Christodoulou, Gkatzelis, and
  Schoepflin]{CGS2022}
Giorgos Christodoulou, Vasilis Gkatzelis, and Daniel Schoepflin.
\newblock Optimal deterministic clock auctions and beyond.
\newblock In \emph{13th Innovations in Theoretical Computer Science Conference
  (ITCS 2022)}. Schloss Dagstuhl-Leibniz-Zentrum f{\"u}r Informatik, 2022.

\bibitem[Cole and Roughgarden(2014)]{ColeR14}
Richard Cole and Tim Roughgarden.
\newblock The sample complexity of revenue maximization.
\newblock In David~B. Shmoys, editor, \emph{Symposium on Theory of Computing,
  {STOC} 2014, New York, NY, USA, May 31 - June 03, 2014}, pages 243--252.
  {ACM}, 2014.
\newblock \doi{10.1145/2591796.2591867}.
\newblock URL \url{https://doi.org/10.1145/2591796.2591867}.

\bibitem[Correa et~al.(2017)Correa, Foncea, Hoeksma, Oosterwijk, and
  Vredeveld]{CFHOV2017}
Jos{\'e} Correa, Patricio Foncea, Ruben Hoeksma, Tim Oosterwijk, and Tjark
  Vredeveld.
\newblock Posted price mechanisms for a random stream of customers.
\newblock In \emph{Proceedings of the 2017 ACM Conference on Economics and
  Computation}, pages 169--186, 2017.

\bibitem[Correa et~al.(2019{\natexlab{a}})Correa, D{\"u}tting, Fischer, and
  Schewior]{correa2019prophet}
Jos{\'e} Correa, Paul D{\"u}tting, Felix Fischer, and Kevin Schewior.
\newblock Prophet inequalities for iid random variables from an unknown
  distribution.
\newblock In \emph{Proceedings of the 2019 ACM Conference on Economics and
  Computation, EC}, pages 3--17. {ACM}, 2019{\natexlab{a}}.

\bibitem[Correa et~al.(2019{\natexlab{b}})Correa, Foncea, Hoeksma, Oosterwijk,
  and Vredeveld]{Cor2019}
Jos{\'e} Correa, Patricio Foncea, Ruben Hoeksma, Tim Oosterwijk, and Tjark
  Vredeveld.
\newblock Recent developments in prophet inequalities.
\newblock \emph{ACM SIGecom Exchanges}, 17\penalty0 (1):\penalty0 61--70,
  2019{\natexlab{b}}.

\bibitem[Correa et~al.(2019{\natexlab{c}})Correa, Saona, and Ziliotto]{CSZ2019}
Jose Correa, Raimundo Saona, and Bruno Ziliotto.
\newblock Prophet secretary through blind strategies.
\newblock In \emph{Proceedings of the Thirtieth Annual ACM-SIAM Symposium on
  Discrete Algorithms}, pages 1946--1961. SIAM, 2019{\natexlab{c}}.

\bibitem[Correa et~al.(2020)Correa, Cristi, Epstein, and Soto]{correa2020two}
Jos{\'e}~R Correa, Andr{\'e}s Cristi, Boris Epstein, and Jos{\'e}~A Soto.
\newblock The two-sided game of googol and sample-based prophet inequalities.
\newblock In \emph{Proceedings of the Fourteenth Annual ACM-SIAM Symposium on
  Discrete Algorithms}, pages 2066--2081. SIAM, 2020.

\bibitem[Daskalakis et~al.(2020)Daskalakis, Fishelson, Lucier, Syrgkanis, and
  Velusamy]{daskalakis2020simple}
Constantinos Daskalakis, Maxwell Fishelson, Brendan Lucier, Vasilis Syrgkanis,
  and Santhoshini Velusamy.
\newblock Simple, credible, and approximately-optimal auctions.
\newblock In \emph{Proceedings of the 21st ACM Conference on Economics and
  Computation}, pages 713--713, 2020.

\bibitem[De~Keijzer et~al.(2020)De~Keijzer, Kyropoulou, and Ventre]{DKV2020}
Bart De~Keijzer, Maria Kyropoulou, and Carmine Ventre.
\newblock Obviously strategyproof single-minded combinatorial auctions.
\newblock In \emph{47th International Colloquium on Automata, Languages, and
  Programming, ICALP 2020}, page~71, 2020.

\bibitem[den Hollander(2008)]{den2008large}
F.~den Hollander.
\newblock \emph{Large Deviations}.
\newblock Fields Institute monographs. American Mathematical Society, 2008.
\newblock ISBN 9780821844359.

\bibitem[D{\"u}tting and Kleinberg(2015)]{DK2015}
Paul D{\"u}tting and Robert Kleinberg.
\newblock Polymatroid prophet inequalities.
\newblock In \emph{Algorithms-ESA 2015}, pages 437--449. Springer, 2015.

\bibitem[D{\"{u}}tting et~al.(2017)D{\"{u}}tting, Gkatzelis, and
  Roughgarden]{DGR2017}
Paul D{\"{u}}tting, Vasilis Gkatzelis, and Tim Roughgarden.
\newblock The performance of deferred-acceptance auctions.
\newblock \emph{Math. Oper. Res.}, 42\penalty0 (4):\penalty0 897--914, 2017.

\bibitem[D{\"u}tting et~al.(2017)D{\"u}tting, Talgam-Cohen, and
  Roughgarden]{DTR2017}
Paul D{\"u}tting, Inbal Talgam-Cohen, and Tim Roughgarden.
\newblock Modularity and greed in double auctions.
\newblock \emph{Games and Economic Behavior}, 105:\penalty0 59--83, 2017.

\bibitem[D{\"u}tting et~al.(2020)D{\"u}tting, Feldman, Kesselheim, and
  Lucier]{DFKL2020}
Paul D{\"u}tting, Michal Feldman, Thomas Kesselheim, and Brendan Lucier.
\newblock Prophet inequalities made easy: stochastic optimization by pricing
  nonstochastic inputs.
\newblock \emph{SIAM Journal on Computing}, 49\penalty0 (3):\penalty0 540--582,
  2020.

\bibitem[Ehsani et~al.(2018)Ehsani, Hajiaghayi, Kesselheim, and
  Singla]{EHKS2018}
Soheil Ehsani, MohammadTaghi Hajiaghayi, Thomas Kesselheim, and Sahil Singla.
\newblock Prophet secretary for combinatorial auctions and matroids.
\newblock In \emph{Proceedings of the Twenty-Ninth Annual ACM-SIAM Symposium on
  Discrete Algorithms}, pages 700--714. SIAM, 2018.

\bibitem[Esfandiari et~al.(2017)Esfandiari, Hajiaghayi, Liaghat, and
  Monemizadeh]{EHLM2017}
Hossein Esfandiari, MohammadTaghi Hajiaghayi, Vahid Liaghat, and Morteza
  Monemizadeh.
\newblock Prophet secretary.
\newblock \emph{SIAM Journal on Discrete Mathematics}, 31\penalty0
  (3):\penalty0 1685--1701, 2017.

\bibitem[Essaidi et~al.(2022)Essaidi, Ferreira, and Weinberg]{EFW22}
Meryem Essaidi, Matheus V.~X. Ferreira, and S.~Matthew Weinberg.
\newblock Credible, strategyproof, optimal, and bounded expected-round
  single-item auctions for all distributions.
\newblock In Mark Braverman, editor, \emph{13th Innovations in Theoretical
  Computer Science Conference, {ITCS} 2022, January 31 - February 3, 2022,
  Berkeley, CA, {USA}}, volume 215 of \emph{LIPIcs}, pages 66:1--66:19. Schloss
  Dagstuhl - Leibniz-Zentrum f{\"{u}}r Informatik, 2022.

\bibitem[Ezra et~al.(2018)Ezra, Feldman, and Nehama]{ezra2018prophets}
Tomer Ezra, Michal Feldman, and Ilan Nehama.
\newblock Prophets and secretaries with overbooking.
\newblock In \emph{Proceedings of the 2018 {ACM} Conference on Economics and
  Computation, {EC}}, pages 319--320. {ACM}, 2018.

\bibitem[FCC(2017)]{fcc}
FCC.
\newblock Broadcast incentive auction.
\newblock
  \url{https://www.fcc.gov/about-fcc/fcc-initiatives/incentive-auctions}, 2017.

\bibitem[Feldman et~al.(2014)Feldman, Gravin, and Lucier]{FGL2014}
Michal Feldman, Nick Gravin, and Brendan Lucier.
\newblock Combinatorial auctions via posted prices.
\newblock In \emph{Proceedings of the Twenty-Sixth Annual ACM-SIAM Symposium on
  Discrete Algorithms}, pages 123--135. SIAM, 2014.

\bibitem[Ferraioli et~al.(2021)Ferraioli, Penna, and Ventre]{ferraioli2021two}
Diodato Ferraioli, Paolo Penna, and Carmine Ventre.
\newblock Two-way greedy: Algorithms for imperfect rationality.
\newblock In \emph{International Conference on Web and Internet Economics},
  pages 3--21. Springer, 2021.

\bibitem[Ferreira and Weinberg(2020)]{ferreira2020credible}
Matheus~VX Ferreira and S~Matthew Weinberg.
\newblock Credible, truthful, and two-round (optimal) auctions via
  cryptographic commitments.
\newblock In \emph{Proceedings of the 21st ACM Conference on Economics and
  Computation}, pages 683--712, 2020.

\bibitem[Gkatzelis et~al.(2017)Gkatzelis, Markakis, and Roughgarden]{GMR2017}
Vasilis Gkatzelis, Evangelos Markakis, and Tim Roughgarden.
\newblock Deferred-acceptance auctions for multiple levels of service.
\newblock In \emph{Proceedings of the 2017 ACM Conference on Economics and
  Computation}, pages 21--38, 2017.

\bibitem[Gkatzelis et~al.(2021)Gkatzelis, Patel, Pountourakis, and
  Schoepflin]{GPPS21}
Vasilis Gkatzelis, Rishi Patel, Emmanouil Pountourakis, and Daniel Schoepflin.
\newblock Prior-free clock auctions for bidders with interdependent values.
\newblock In Ioannis Caragiannis and Kristoffer~Arnsfelt Hansen, editors,
  \emph{Algorithmic Game Theory - 14th International Symposium, {SAGT} 2021,
  Aarhus, Denmark, September 21-24, 2021, Proceedings}, volume 12885 of
  \emph{Lecture Notes in Computer Science}, pages 64--78. Springer, 2021.

\bibitem[Hajiaghayi et~al.(2007)Hajiaghayi, Kleinberg, and Sandholm]{HKS2007}
Mohammad~Taghi Hajiaghayi, Robert Kleinberg, and Tuomas Sandholm.
\newblock Automated online mechanism design and prophet inequalities.
\newblock In \emph{Proceedings of the 22nd National Conference on Artificial
  Intelligence - Volume 1}, AAAI 2007, pages 58--65. AAAI Press, 2007.
\newblock ISBN 9781577353232.

\bibitem[Hartline and Roughgarden(2009)]{HR2009}
Jason~D Hartline and Tim Roughgarden.
\newblock Simple versus optimal mechanisms.
\newblock In \emph{Proceedings of the 10th ACM Conference on Electronic
  Commerce}, pages 225--234. ACM, 2009.

\bibitem[Huang et~al.(2018)Huang, Mansour, and Roughgarden]{HuangMR18}
Zhiyi Huang, Yishay Mansour, and Tim Roughgarden.
\newblock Making the most of your samples.
\newblock \emph{{SIAM} J. Comput.}, 47\penalty0 (3):\penalty0 651--674, 2018.
\newblock \doi{10.1137/16M1065719}.
\newblock URL \url{https://doi.org/10.1137/16M1065719}.

\bibitem[Kagel et~al.(1987)Kagel, Harstad, and Levin]{KHL1987}
John~H Kagel, Ronald~M Harstad, and Dan Levin.
\newblock Information impact and allocation rules in auctions with affiliated
  private values: A laboratory study.
\newblock \emph{Econometrica: Journal of the Econometric Society}, pages
  1275--1304, 1987.

\bibitem[Kim(2015)]{K2015}
Anthony Kim.
\newblock Welfare maximization with deferred acceptance auctions in
  reallocation problems.
\newblock In \emph{Algorithms-ESA 2015}, pages 804--815. Springer, 2015.

\bibitem[Kleinberg and Weinberg(2012)]{KW2012}
Robert Kleinberg and Seth~Matthew Weinberg.
\newblock Matroid prophet inequalities.
\newblock In \emph{Proceedings of the Forty-Fourth Annual ACM Symposium on
  Theory of Computing}, pages 123--136. ACM, 2012.

\bibitem[Krengel and Sucheston(1977)]{KS1977}
Ulrich Krengel and Louis Sucheston.
\newblock Semiamarts and finite values.
\newblock \emph{Bulletin of the American Mathematical Society}, 83\penalty0
  (4):\penalty0 745--747, 1977.

\bibitem[Krengel and Sucheston(1978)]{KS1978}
Ulrich Krengel and Louis Sucheston.
\newblock On semiamarts, amarts, and processes with finite value.
\newblock \emph{Advances in Probability and Related Topics}, 4:\penalty0
  197--266, 1978.

\bibitem[Li(2017)]{L2017}
Shengwu Li.
\newblock Obviously strategy-proof mechanisms.
\newblock \emph{American Economic Review}, 2017.

\bibitem[Loertscher and Marx(2020)]{LM2020}
Simon Loertscher and Leslie~M Marx.
\newblock Asymptotically optimal prior-free clock auctions.
\newblock \emph{Journal of Economic Theory}, page 105030, 2020.

\bibitem[Lucier(2017)]{Luc2017}
Brendan Lucier.
\newblock An economic view of prophet inequalities.
\newblock \emph{ACM SIGecom Exchanges}, 16\penalty0 (1):\penalty0 24--47, 2017.

\bibitem[Milgrom and Segal(2014)]{MS2014}
Paul Milgrom and Ilya Segal.
\newblock Deferred-acceptance auctions and radio spectrum reallocation.
\newblock In Moshe Babaioff, Vincent Conitzer, and David~A. Easley, editors,
  \emph{{ACM} Conference on Economics and Computation, {EC} '14, Stanford , CA,
  USA, June 8-12, 2014}, pages 185--186. {ACM}, 2014.

\bibitem[Milgrom and Segal(2019)]{MS2019}
Paul Milgrom and Ilya Segal.
\newblock Clock auctions and radio spectrum reallocation.
\newblock \emph{Journal of Political Economy}, 2019.

\bibitem[Rubinstein(2016{\natexlab{a}})]{R2016}
Aviad Rubinstein.
\newblock Beyond matroids: Secretary problem and prophet inequality with
  general constraints.
\newblock In \emph{Proceedings of the Forty-Eighth Annual ACM Symposium on
  Theory of Computing}, pages 324--332, 2016{\natexlab{a}}.

\bibitem[Rubinstein(2016{\natexlab{b}})]{R2016Simple}
Aviad Rubinstein.
\newblock On the computational complexity of optimal simple mechanisms.
\newblock In \emph{Proceedings of the 2016 ACM Conference on Innovations in
  Theoretical Computer Science}, pages 21--28. ACM, 2016{\natexlab{b}}.

\bibitem[Rubinstein and Singla(2017)]{RS2017}
Aviad Rubinstein and Sahil Singla.
\newblock Combinatorial prophet inequalities.
\newblock In \emph{Proceedings of the Twenty-Eighth Annual ACM-SIAM Symposium
  on Discrete Algorithms}, SODA 2017, pages 1671--1687, USA, 2017. Society for
  Industrial and Applied Mathematics.

\bibitem[Rubinstein et~al.(2020)Rubinstein, Wang, and Weinberg]{RWW2020}
Aviad Rubinstein, Jack~Z Wang, and S~Matthew Weinberg.
\newblock Optimal single-choice prophet inequalities from samples.
\newblock In \emph{11th Innovations in Theoretical Computer Science Conference
  (ITCS 2020)}. Schloss Dagstuhl-Leibniz-Zentrum f{\"u}r Informatik, 2020.

\bibitem[Vondr{\'{a}}k(2010)]{Vondrak10}
Jan Vondr{\'{a}}k.
\newblock A note on concentration of submodular functions.
\newblock \emph{CoRR}, abs/1005.2791, 2010.
\newblock URL \url{http://arxiv.org/abs/1005.2791}.

\bibitem[Yan(2011)]{Y2011}
Qiqi Yan.
\newblock Mechanism design via correlation gap.
\newblock In \emph{Proceedings of the Twenty-Second Annual ACM-SIAM Symposium
  on Discrete Algorithms}, pages 710--719. SIAM, 2011.

\end{thebibliography}
\end{document}